\pdfoutput=1
\documentclass[a4paper,10pt]{article}

\usepackage[bbgreekl]{mathbbol}
\usepackage{mathrsfs}
\usepackage{graphicx}
\usepackage{subfigure}
\usepackage{amsmath}
\usepackage{amsfonts}
\usepackage{amssymb}
\usepackage{amsthm}
\usepackage{multirow}
\usepackage{float}
\usepackage{xypic}
\usepackage[normalem]{ulem}

\newtheorem{theorem}{Theorem}[section]

\newtheorem{proposition}{Proposition}[section]

\newtheorem{algorithm}{Algorithm}[section]

\newcommand{\AN}{{\mathcal A}_N}
\newcommand{\R}{{\mathbb R}}
\newcommand{\Z}{{\mathbb Z}}

\newcommand{\RN}{{\mathcal R}_N}
\newcommand{\rhobar}{{\overline{\rho}}}

\title{Pair densities in density functional theory}
\author{Huajie Chen\thanks{Zentrum Mathematik, Technische Universit\"{a}t M\"{u}nchen,
Boltzmannstra{\ss}e 3, 85747 Garching, Germany.
E-mail: {\tt chenh@ma.tum.de}.}
$\;$and Gero Friesecke\thanks{Zentrum Mathematik, Technische Universit\"{a}t M\"{u}nchen,
Boltzmannstra{\ss}e 3, 85747 Garching, Germany.
E-mail: {\tt gf@ma.tum.de}.}}
\date{}

\begin{document}
\maketitle

\begin{abstract}
The exact interaction energy of a many-electron system is determined by the
electron pair density, which is not well-approximated in standard Kohn-Sham density functional models. Here we study the (complicated but well-defined) exact universal map from density to pair density. We survey how many common functionals, including the most basic version of the LDA (Dirac exchange with no correlation contribution), arise from particular approximations of this map. We develop an algorithm to compute the map numerically, and apply it to  one-parameter families $\{\alpha\rho(\alpha x)\}_{\alpha>0}$ of one-dimensional homogeneous and inhomogeneous single-particle densities. We observe that the pair density develops remarkable multiscale patterns which strongly depend on both the particle number and the ``width'' $\alpha^{-1}$ of the single-particle density. The simulation results are confirmed by rigorous asymptotic results in the limiting regimes $\alpha>>1$ and $\alpha<<1$. 
For one-dimensional homogeneous systems, we show that the whole spectrum of patterns is reproduced surprisingly well by a simple asymptotics-based ansatz which slowly smoothens out the `strictly correlated' $\alpha=0$ pair density while slowly turning on the $\alpha=\infty$ `exchange' terms as $\alpha$ increases.
Our findings lend theoretical support to the celebrated semi-empirical idea \cite{becke93} to mix in a fractional amount of exchange, albeit not to assuming the mixing to be additive and taking the fraction to be a system-independent constant. 
\end{abstract}

\section{Introduction}\label{sec-introduction}\setcounter{equation}{0}

Density functional theory (DFT) \cite{hohenberg64,kohn65,PY89} provides the most widely used models for computing ground state electronic energies and densities in chemistry, materials science, biology, and nanosciences.
The success of DFT lies in the use of exchange-correlation functionals that model the intricate many-body interaction energy by explicit expressions in terms of the one-body density or the one-body Kohn-Sham orbitals.   
Although currently available approximations, such as B3LYP \cite{becke93,lee88} or PBE \cite{PBE}, perform remarkably well for a wide range of systems, the DFT models 
exhibit well-known failures when strong correlation effects are present, as arising for example in the breaking of chemical bonds \cite{cohen12}. Therefore, finding an accurate single-particle formalism that remains reliable in strongly correlated regimes remains a major challenge. 
\\[2mm]
In this paper we shed new light on this challenge by studying the exact map 
$\rho\mapsto\rho_2$ from single-particle density to {\it pair density} whose existence is assured by abstract DFT.
The exact interaction energy is obtained by integrating the pair density against the Coulomb repulsion potential (see \eqref{v_ee_rho2} below), so any approximation $\rho\mapsto\tilde{\rho}_2$ yields an approximate interaction energy functional. We take the view, first advocated by Gunnarsson and Lundqvist \cite{GunLun}, that the exact density-to-pair-density map $\rho\mapsto\rho_2$ is a better starting point to understand or design model interaction energy functionals than the commonly used density-to-interaction-energy map $\rho\mapsto V_{ee}[\rho]$. This is because the pair density, a function on two-body configuration space, encodes a wealth of physically and mathematially interesting information about a many-body quantum system which is ``averaged out'' in the interaction energy, a mere number. In particular, comparing 
exact and approximate pair densities does not just yield a total interaction energy error, but also reveals where in two-body configuration space the error is localized.
\\[2mm]
The main results in our paper are careful simulations of the exact density-to-pair-density map for typical one-dimensional model systems, with different electron numbers and different density profiles varying from ``concentrated'' to ``dilute''. The algorithm we develop for this purpose allows one to deal with the infinite-dimensional, nonlinear constraint of fixed single-particle density which appears in the definition of the map. To obtain a simple form of this constraint after discretization, we use a finite-element basis from computational mathematics instead of the usual basis sets of quantum chemistry. We observe remarkable multi-scale patterns in the pair density which strongly depend on both the particle number and the ``width'' of the density profile. See e.g. Figure \ref{fig-fermion-convex} in Section \ref{sec-fermion}. These patterns are not accurately captured (except in extreme regimes) by any of the currently used DFT models, but in our view constitute fundamental low-dimensional manifestations of exact DFT. We thus hope that our simulations, despite the limitation to one-dimensional model densities, offer exciting glimpses of possible future DFT models. 
\\[2mm]
In the remainder of this Introduction we informally discuss the definition of the exact density-to-pair density map, display two instructive extreme and opposite approximations in the DFT literature, and explain how our simulations seamlessly connect all three. The exact density-to-pair-density map is constructed as follows:
\begin{equation} \label{rho2_exact}
   \rho \mapsto \Psi \mapsto \rho_2
\end{equation}
where $\rho\mapsto\Psi$ is the map obtained by Levy-Lieb constrained search \cite{levy79,lieb83}, i.e. $\Psi$ is the $N$-electron wavefunction which minimizes kinetic plus potential energy, $T[\Psi]+V_{ee}[\Psi]$, subject to the constraint that $\Psi$ has single-particle density $\rho$ (see Section 2 for notation, function spaces, and further explanation), and the pair density associated to any $N$-electron wavefunction $\Psi$ is
\begin{equation}\label{rho2_def}
\rho_2^\Psi(x,y) = \binom{N}{2}\int_{(\mathbb{R}^3)^{N-2}}\sum_{\sigma_1,\cdots,\sigma_N\in\mathbb{Z}_2}
|\Psi(x,\sigma_1,y,\sigma_2,x_3,\sigma_3,\cdots,x_N,\sigma_N)|^2 dx_3\cdots dx_N.
\end{equation}
Here $(x_i,\sigma_i)\in\mathbb{R}^3\times\mathbb{Z}_2$
are space-spin coordinates for the $i^{th}$ electron. Minimizing $\Psi$'s always exist \cite{lieb83}, and the complication of possible non-uniqueness is discussed in Section \ref{sec-universalmap}. The electron-electron interaction energy is a simple explicit functional of the pair density,
\begin{equation}\label{v_ee_rho2}
V_{ee}[\Psi]=\int_{\mathbb{R}^6}\frac{\rho_2^{\Psi}(x,y)}{|x-y|}dx\, dy.
\end{equation}
Hence any approximate expression of the pair density in terms of the single-particle density gives, by substitution into \eqref{v_ee_rho2}, an approximate interaction energy functional. Numerous functionals have been formulated in this way \cite{GunLun, becke87, PerdewWang, perdew95, PerdewBurkeWang}.   
Standard DFT models start from a {\it statistical independence ansatz}
\begin{equation}\label{rho2_mean_intro}
\rho_2(x,y)=\frac{1}{2}\rho(x)\rho(y),
\end{equation}
and include all the many-body effects in a correcting exchange-correlation energy functional.
Opposite to this uncorrelated ansatz, there is the more recent {\it strictly correlated electrons} (SCE) model \cite{seidl99a,seidl99b,seidl07}, which is attracting attention in the mathematics literature \cite{cotar13a,friesecke13,colombo,Ghoussoub} due to its connection with optimal transportation theory and which arises from neglecting the kinetic energy in the constrained search in \eqref{rho2_exact}. The corresponding ansatz for the pair density is
\begin{equation}\label{rho2_sce_intro}
\rho_2(x,y)=\frac{1}{2N}\sum_{i\neq j}\int_{\mathbb{R}^3}\rho(z)\delta(x-T_i(z))\delta(y-T_j(z))dz
\end{equation}
with $T_i:\mathbb{R}^3\rightarrow\mathbb{R}^3,~i=1,\cdots,N$ being certain optimal transport maps, see Section \ref{sec:approx} for more details.

The formalisms \eqref{rho2_mean_intro} and \eqref{rho2_sce_intro} give ``extreme'' pair densities. True pair densities, unlike \eqref{rho2_mean_intro}, are expected to localize in certain regions due to shell structure or ionicity avoidance; but \eqref{rho2_sce_intro} emphasizes this localization too much and misses its quantum features.  
Most of the practically interesting models, such as the local density approximation (see Section \ref{sec:approx}), lie ``inbetween'' these two extreme distributions. But what kind of ``interpolation'' is the right one, and captures true pair densities \eqref{rho2_exact}?

A marvellous tool to approach this question is {\it density scaling}, introduced in the context of exact DFT by Levy and Perdew \cite{PerdewLevy}, which is closely related to the {\it adiabatic connection} utilized in many DFT studies (see e.g. \cite{GunLun,PerdewLevy,seidl99b,seidl07,becke14}). 
Alongside a given density $\rho \, : \, \R^d\to\R$, consider -- as we shall in our simulations -- its re-scalings
\begin{equation} \label{scale_intro}
           (D_\alpha\rho)(x) = \alpha^d \rho(\alpha x), \;\;\;\alpha\in(0,\infty).
\end{equation}
The parameter $\alpha$ seamlessely rescales a dilute system ($\alpha<<1$) into a concentrated one ($\alpha>>1$). But the associated pair densities do {\it not} just change by a rescaling,
that is to say $\rho_2[D_\alpha\rho]\neq D_\alpha\rho_2[\rho]$ (where $D_\alpha$ acts on pair densities as $(D_\alpha\rho_2)(x,y)=\alpha^{2d}\rho_2(\alpha x,\alpha y)$). 
Instead, it follows from the arguments in \cite{PerdewLevy} that the following diagram commutes:
\begin{equation} \label{pdscale_intro}
\begin{xy}
  \xymatrix{{\rho} \ar[r]^{{\rm scale}} 
            \ar[d]_{\min \, {\alpha} T + V_{ee}} & 
                        {D_{\alpha}\rho} \ar[d]^{\min \, T + V_{ee}} \\
            {\;\;D_{\alpha^{-1}}\rho_2[D_\alpha\rho]\;\;}  &
                {\;\;\;\rho_2[D_\alpha\rho]\;\;\;} \ar[l]_{\;\;\;\;{\rm scale}\;{\rm back}}
                } 
\end{xy}
\end{equation}
Here `min' means find the minimizing wavefunction under the constraint of the given one-body density and take the resulting pair density.
See Proposition \ref{propo-Falpha-1} below. Thus the scaling parameter $\alpha$ in 
\eqref{scale_intro} acts as a coupling constant in the one-parameter family of variational problems on the left which govern the pair density. This family ``adiabatically'', i.e. while keeping the density fixed, connects the problem of minimizing just $V_{ee}$ ($\alpha=0$) via $T+V_{ee}$ ($\alpha=1$) to minimizing just $T$ ($\alpha=\infty$).\footnote{Andreas Savin suggested to us the name {\it two-sided adiabatic connection} because it combines the classical connection to $T$ (in our parametrization, $1/\alpha\in [0,1]$) with the more recent one to $V_{ee}$ \cite{seidl07} ($\alpha\in[0,1]$).} Nontrivial but well known formal asymptotics for the minimizing wavefunction for $\alpha\to 0$ \cite{seidl99a} and $\alpha\to\infty$ (see e.g. \cite{becke14})
together with formula \eqref{rho2_def} then suggests the following: the true pair density is asymptotic to that of the SCE state, eq. \eqref{rho2_sce_intro}, as $\alpha\to 0$ \cite{seidl99a}, and to that of the Slater determinant of the Kohn-Sham orbitals as $\alpha\to\infty$. The latter reduces to \eqref{rho2_mean_intro} when $N=2$ or when the particles are bosons, and in addition contains `exact exchange' (see Section \ref{sec:approx}) for higher $N$. See Sections \ref{sec-hk}, \ref{sec:asy} for more details and rigorous proofs in special cases. 

Now back to our central question: which ``interpolation'' between the extreme pair densities \eqref{rho2_mean_intro} and \eqref{rho2_sce_intro} is right? Our numerical results for 
typical families \eqref{scale_intro} of one-dimensional densities with different particle numbers show that {\it many different interpolations are right}. See, e.g., Figure \ref{fig-fermion-convex}. The true pair densities form a two-parameter family which strongly depend on both the particle number $N$ and the scaling parameter $\alpha$. At fixed $N$ they steadily ``cross over'' from \eqref{rho2_mean_intro} (plus exact exchange when $N>2$) to \eqref{rho2_sce_intro}. The impractical, highly implicit definition \eqref{rho2_def} is able to pick out the right parameter values from the density, but {\it simple explicit formulae will not}.  
In the very special case of homogeneous systems in one dimension (see Figure \ref{fig-fermion-periodic}) we design an ansatz which does. The idea is to simultaneously smoothen out the pair densities from the strongly interacting limit and fading out the exchange terms from the weakly interacting limit. But the correct smoothing lengthscale and the correct fraction of exchange keep changing with $N$ and $\alpha$. See Table \ref{table-appro-fermion}. Our simulations to some extent support the celebrated idea \cite{becke93}\footnote{According to a recent article in {\it Nature} (29.10.2014), one of the Top Ten most highly cited scientific papers of all time, and the most highly cited one written after 1990.} underlying the functional B3LYP to mix in a fraction of exact exchange. But they show that the right fraction, taken to be $0.2$ in B3LYP \cite{becke93}, is in fact {\it not} constant. At present we have no proposal how the right fraction could be adaptively picked out in realistic (inhomogeneous, 3D) simulations. 


The remainder of this paper is arranged as follows.
In the next two sections, we recall basic aspects of DFT and give the precise definition of the universal density-to-pair-density map. We then show in Section \ref{sec:approx} how some common DFT functionals arise from approximations of this map.
Section \ref{sec-numerical} describes our numerical simulations of the true pair densities
of homogeneous and inhomogeneous one dimensional systems for both bosons and fermions. In Section \ref{sec:asy} we present rigorous asymptotic results which confirm the numerical findings. In Section \ref{sec-ansatz}, we propose an ansatz for approximating pair densities of one dimensional homogeneous electron systems.
Finally, in Section \ref{sec-conclusion} we give conclusions and some future perspectives.

\section{Density functional theory}\label{sec-pairdensity}
\setcounter{equation}{0}\setcounter{figure}{0}

Here we recall the basic functionals of DFT which will be needed in the following. A standard reference is \cite{PY89}. Readers familiar with DFT might want to skip this section. We consider a general system of $N$ nonrelativistic electrons in $\R^d$ under the influence of an external potential $v_{ext} \, : \, \R^d\to\R$ and a repulsive pair potential $v_{ee}\, : \, \R^d\to\R$. Prototypically, for real physical systems, 
$$
   d=3, \;\;\; v_{ee}(x-y)=\frac{1}{|x-y|} \;\;(x,y\in\R^d),
$$
and $v_{ext}$ is the electrostatic potential generated by $M$ nuclei of charges $Z_1,..,Z_M>0$ located at positions $R_1,..,R_M\in\R^3$,
\begin{eqnarray*}
v_{ext}(x)=-\sum_{I=1}^{M}\frac{Z_I}{|x-R_I|}\quad (x\in\mathbb{R}^3).
\end{eqnarray*}
The quantum mechanical ground state energy of the system is given by 
\begin{equation}\label{GSE}
    E_0 = \inf_{\Psi\in\AN} \Bigl( T[\Psi] + V_{ee}[\Psi] + V[\Psi]\Bigr),
\end{equation}
where $\AN$ is the following class of admissible wavefunctions
\begin{equation} \label{AN}
\AN = \left\{ \Psi
\in L^2((\mathbb{R}^d\times\mathbb{Z}_2)^N),
~\nabla\Psi\in L^2, ~\Psi~{\rm antisymmetric},~\|\Psi\|_{L^2}=1 \right\},
\end{equation}
$\mathbb{Z}_2=\{\uparrow,\downarrow\}$, 
and $T$, $V_{ee}$, $V$ are the following functionals: 
\begin{equation}\label{eq-T}
T[\Psi] = \frac{1}{2}\int_{(\mathbb{R}^d)^N}\sum_{\sigma_1,\cdots,\sigma_N\in\mathbb{Z}_2}
\sum_{i=1}^N|\nabla_{x_i}\Psi(x_1,\sigma_1,\cdots,x_N,\sigma_N)|^2
dx_1\cdots dx_N
\end{equation}
(kinetic energy),
\begin{equation}\label{eq-ee}
V_{ee}[\Psi] = \int_{(\mathbb{R}^d)^N}\sum_{\sigma_1,\cdots,\sigma_N\in\mathbb{Z}_2}
\sum_{1\leq i<j\leq N} v_{ee}(x_i-x_j)|\Psi(x_1,\sigma_1,\cdots,x_N,\sigma_N)|^2 dx_1\cdots dx_N
\end{equation}
(electron-electron interaction energy), and
\begin{equation}
V[\Psi] = \int_{(\mathbb{R}^d)^N}\sum_{\sigma_1,\cdots,\sigma_N\in\mathbb{Z}_2}
     \sum_{i=1}^N v_{ext}(x_i)|\Psi(x_1,\sigma_1,\cdots,x_N,\sigma_N)|^2 dx_1\cdots dx_N
\end{equation}
(external potential energy). 
\\[2mm]
A central result of DFT going back to Hohenberg and Kohn is the following. We state the result here in the form discovered by M.Levy \cite{levy79} and made rigorous by \cite{lieb83}. The quantum mechanical ground state energy \eqref{GSE} can be recovered exactly by minimizing a certain density functional,
\begin{equation} \label{HKthm}
   E_0 = \inf_{\rho\in\RN} \Bigl( F_{HK}[\rho] + \int_{\R^d} v_{ext}\rho\Bigr),
\end{equation}
where
\begin{equation}\label{FHK}
F_{\rm HK}[\rho] = \min_{\Psi\in\AN,\Psi\mapsto\rho} \left\{ T[\Psi] + V_{ee}[\Psi] \right\}.
\end{equation}
Here $\Psi\mapsto\rho$ means that $\Psi$ has single-particle density $\rho$, i.e.
\begin{equation}\label{single-density}
\rho(x) = N\int_{(\mathbb{R}^3)^{N-1}}\sum_{\sigma_1,\cdots,\sigma_N\in\mathbb{Z}_2}
|\Psi(x,\sigma_1,x_2,\sigma_2,\cdots,x_N,\sigma_N)|^2 dx_2\cdots dx_N ,
\end{equation}
and $\RN$ is the space of densities arising via \eqref{single-density} from wavefunctions $\Psi\in\AN$.
Note that the space $\RN$ of densities is known explicitly: by a result of Lieb \cite{lieb83},
\begin{equation}\label{RN}
\RN=\left\{\rho\, : \, \R^d\to\R\, | \, \rho\geq 0, \; \sqrt{\rho}\in H^1(\R^d),\, \int_{\R^d}\rho=N \right\}.
\end{equation}
We also note that $F_{\rm HK}$ is a universal functional of $\rho$, in the sense that it does not depend on the external potential $v_{ext}$. 
Minimizers in \eqref{FHK} always exist provided $\rho\in\mathcal{R}_N$ \cite{lieb83}.

The complexity of the DFT model \eqref{HKthm} lies in that no tractable
expression for $F_{\rm HK}$ is known that could be used in numerical simulations.
In practice, $F_{\rm HK}[\rho]$ is approximated by the sum of a kinetic part and an interaction part,
\begin{equation}\label{add}
F_{\rm HK}[\rho] \approx \tilde{T}[\rho]+\tilde{V}_{ee}[\rho],
\end{equation}
leading to an approximate expression for the ground state energy,
\begin{equation}\label{GSEapprox}
    E_0 \approx \tilde{E}_0 = \inf_{\rho\in\RN} \Bigl( \tilde{T}[\rho] + \tilde{V}_{ee}[\rho] + \int_{\R^d} v_{ext}\rho\Bigr).
\end{equation}
Many clever and useful approximate functionals $\tilde{T}$ and $\tilde{V}_{ee}$ have  been proposed and utilized to simulate a wide range of systems (see e.g. \cite{PY89,becke14}). Particularly fruitful has been the idea of Kohn and Sham \cite{kohn65} to construct a kinetic energy functional $\tilde{T}$ with the help of single-particle orbitals of a non-interacting reference system:
\begin{multline}\label{T-KS}
\tilde{T}[\rho] = T_{\rm KS}[\rho] = \min\left\{ \frac{1}{2}\sum_{i=1}^{N}\int|\nabla\phi_i|^2,
~ \phi_i\in H^1(\mathbb{R}^d\times\mathbb{Z}_2), \right. \\
\left. ~\int\overline{\phi_i}\phi_j=\delta_{ij},~
\sum_{i=1}^{N}\sum_{\sigma\in\mathbb{Z}_2}|\phi_i(x,\sigma)|^2=\rho(x) \right\},
\quad\quad
\end{multline}
where for any function $f=f(x,\sigma)$ of $(x,\sigma)\in\R^d\times\Z_2$ we use the abbreviation
$\int f = \sum_{\sigma\in\mathbb{Z}_2}\int_{\mathbb{R}^d}f(x,\sigma)dx$. It is easy to see that $\tilde{T}$ is the same as the functional obtained by omitting $V_{ee}$ in \eqref{FHK} and restricting the minimization to Slater determinants
\begin{equation}\label{slater}
  \Psi(x_1,\sigma_1,..,x_N,\sigma_N) =  \frac{1}{\sqrt{N!}}
  \left|\begin{array}{cccc}
\varphi_{1}(x_1,\sigma_1) & \cdots & \varphi_{N}(x_1,\sigma_1) \\[1ex]
\vdots & \ddots & \vdots \\[1ex]
\varphi_{1}(x_N,\sigma_N) & \cdots & \varphi_{N}(x_N,\sigma_N)
\end{array}\right| .
\end{equation}
Minimizers in \eqref{T-KS} always exist provided $\rho\in\mathcal{R}_N$ \cite{lieb83}.
\\[2mm]
Using \eqref{T-KS} as the kinetic energy in \eqref{add}, as is done in almost all simulations to date, and the orbitals $\Phi=(\phi_1,\cdots,\phi_N)$
as the basic variable, the ground state energy of the system becomes
\begin{equation}\label{min-ks}
E_0 \approx \tilde{E}_0 = \inf \left\{ \frac{1}{2}\sum_{i=1}^{N}\int|\nabla\phi_i|^2
+ \int_{\mathbb{R}^d}v_{ext}\rho_{\Phi} + \tilde{V}_{ee}[\rho_{\Phi}], ~~
\phi_i\in H^1(\mathbb{R}^3),~\int \overline{\phi_i}\phi_j=\delta_{ij} \right\}
\end{equation}
with the single-particle density 
$$
  \rho_{\Phi}(x) = \sum_{s\in\Z_2}\sum_{i=1}^{N}|\phi_i(x,s)|^2.
$$
The remaining problem, and the one of interest to us, is to design accurate approximations for $\tilde{V}_{ee}[\rho_{\Phi}]$.

\section{Universal density-to-pair-density map}\label{sec-universalmap}
\setcounter{equation}{0}\setcounter{figure}{0}
Our starting point for looking at interaction energy functionals will be the universal, exact density to pair density map delivered by abstract DFT. 
Following Levy \cite{levy79} this map is defined as follows. Recall from \eqref{rho2_def} that $\rho_2^\Psi$ denotes the pair density of the wavefunction $\Psi$. 
\\[2mm]
{\bf Definition} (Universal density to pair density map) {\it For any one-body density $\rho$ of an $N$-electron system, that is to say for any $\rho$ belonging to the class $\RN$ in \eqref{RN}, } 
\begin{equation} \label{UPD}
  \rho_2[\rho] = \{ \rho_2^\Psi \, | \, \Psi\in\AN \mbox{ is a minimizer of }T+V_{ee} \mbox{ subject to }\Psi\mapsto\rho\}. 
\end{equation}
Just like the map $\rho\mapsto F_{HK}[\rho]$, the map $\rho\mapsto\rho_2[\rho]$ is universal, i.e. independent of the external potential. The above definition requires, and it was proved mathematially by Lieb \cite{lieb83}, that a minimizing $\Psi$ exists. Note however that the minimizer may not be unique. Hence the map is possibly multi-valued, that is to say $\rho_2[\rho]$ is possibly a {\it set} of pair densities rather than a single pair density. Simple explicit examples of nonuniqueness in the case when $T+V_{ee}$ is replaced by $T$ are given in Section \ref{sec:asy}.
\\[2mm]
The physical significance of $\rho_2[\rho]$ comes from the following direct consequence of formulae \eqref{GSE}, \eqref{HKthm}: if $\Psi$ is any exact quantum mechanical ground state, i.e. a minimizer of the right hand side of \eqref{GSE} for some external potential $v_{ext}$, and $\Psi$ has one-body density $\rho$, then $\rho_2[\rho]$ is the exact pair density of $\Psi$, and the functional 
\begin{equation}\label{Veeex}
   \overline{V}_{ee}[\rho] := \int_{\R^d} v_{ee}(x-y)\rho_2[\rho](x,y) \, dx \, dy
\end{equation}
agrees with the exact interaction energy $V_{ee}[\Psi]$ from \eqref{eq-ee}. 
\section{Approximate density-to-pair-density maps} \label{sec:approx}
It is obvious that substituting any approximation $\tilde{\rho}_2[\rho]$ of the density to pair density map $\rho_2[\rho]$ into \eqref{Veeex} yields an approximate interaction energy functional $\tilde{V}_{ee}[\rho]$. Conversely, we now show that many basic approximate functionals used in practice can be derived in this way. 
In some cases, such as the bare Hartree functional or `exact exchange' (Examples 1 and 3), this is trivial. For the LDA in its most basic form (Dirac exchange with no correlation contribution, Example 2) it is not, and we are not aware that an {\it exact} equivalence to a pair density model for any inhomogeneous density as given below has been stated previously, even though good approximate pair density formulations are well known \cite{GunLun}. For interesting work relating advanced DFT functionals to pair density approximations we refer to \cite{becke87, PerdewWang, PerdewBurkeWang}. 
\\[2mm]
{\bf Example 1.} (statistical independence)
The simplest idea is to assume statistical independence,
\begin{eqnarray}\label{mean-field}
\tilde{\rho}_2[\rho](x,y) = \frac{1}{2}\rho(x)\rho(y).
\end{eqnarray}
Substituting this density to pair density map into the right-hand side of \eqref{Veeex} leads to
the Hartree functional
\begin{equation} \label{hartree}
\tilde{V}_{ee}[\rho] = \frac{1}{2}\int_{\mathbb{R}^6}\frac{\rho(x)\rho(y)}{|x-y|} dx\, dy.
\end{equation}
While never used on its own, together with some correcting exchange-correlation functional $E_{xc}[\rho]$ it is contained in virtually all DFT models, including state of the art ones like B3LYP \cite{lee88, becke93} or PBE \cite{PBE}. 
\\[2mm]
{\bf Example 2.} (Local density approximation with Dirac exchange) 
For the free (i.e., noninteracting) electron gas, the pair density can be determined explicitly (see e.g. \cite{PY89} and, for a mathematical account, \cite{friesecke97}).
In this case the single-particle density is a constant, $\rho(x)\equiv\bar{\rho}$, and the
pair density is
\begin{equation}\label{rho2-free}
  \rho_2(x,y) = \frac{1}{2}\bar{\rho}^2 -\frac{1}{4}\bar{\rho}^2 h^2((3\pi^2\bar{\rho})^{1/3}|x-y|),
\end{equation}
where $h(s) = 3(\sin s-s\cos s)/s^3$. We claim that the inhomogeneous version 
\begin{equation}\label{rho2-nonunif}
  \tilde{\rho}_2[\rho](x,y) = \frac12 \rho(x)\rho(y) - \frac18 \rho(x)^2 h^2((3\pi^2\rho(x))^{1/3}|x-y|)
     - \frac18 \rho(y)^2 h^2((3\pi^2\rho(y))^{1/3}|x-y|)
\end{equation}
yields the interaction energy
\begin{equation}\label{Vee-Xalpha}
\tilde{V}_{ee}[\rho] = \frac{1}{2}\int_{\mathbb{R}^6}\frac{\rho(x)\rho(y)}{|x-y|}dx\, dy
-c_x\int_{\mathbb{R}^3}\rho(x)^{4/3}dx
\end{equation}
with constant $c_x=\frac{3}{4}(\frac{3}{\pi})^{1/3}$. This can be seen as follows. For each of the non-mean-field terms, just integrate out the variable not contained in the argument of $h$, e.g., using spherical polar coordinates for $y$ centered at $x$ and abbreviating $k_F(x)=(3\pi^2\rho(x))^{1/3}$,
$$
     \int_{\mathbb{R}^3}h^2(k_F(x)|x-y|) \, dy = 4\pi \int_0^\infty h^2(k_F(x)r)r^2 dr = \frac{4\pi}{(3\pi^2\rho(x))^{2/3}} \int_0^\infty h^2(r'){r'}^2 dr', 
$$
and determine the remaining one-dimensional integal as in the discussion of the homogeneous case in \cite{PY89, friesecke97}.
Eq. \eqref{Vee-Xalpha}
is the simplest of the local density approximations (LDA) \cite{kohn65,PY89,martin05}. The second term of \eqref{Vee-Xalpha} is the celebrated Dirac exchange functional \cite{dirac30}. We remark that from Dirac's original derivation it is not clear how to relate this functional to the pair density as he used a semiclassical limit argument for the (one-body) energy density per unit volume. 

Strange as the model \eqref{rho2-nonunif} for the pair density may look, it provides a precise way to state what the LDA really does: the pair density is assumed to be independent at long range (note that $h(r)$ goes to zero as $r$ gets large), while at short range it contains an ``exchange hole'' \footnote{see e.g. \cite{martin05} for more information about this semi-empirical notion} of fixed shape coming from free electron gas theory whose diameter is of order $\rho(x)^{-1/3}$. 
\\[3mm]
{\bf Example 3.} (exact exchange)
To obtain ``exact'', i.e. Hartree-Fock-like, exchange \cite{becke93}, one takes
\begin{equation} \label{exex}
  \tilde{\rho}_2[\rho](x,y) = \rho_2^{\Psi}(x,y),
\end{equation}
where $\Psi$ is the Slater determinant \eqref{slater} composed of the ($\rho$-dependent) minimizing orbitals  $\varphi_1,...,\varphi_N$ in the definition of the Kohn-Sham kinetic energy functional \eqref{T-KS}. A more explicit expression for $\tilde{\rho}_2$ is obtained by using the well known expression for the pair density of a Slater determinant (see e.g. \cite{helgaker00}):
\begin{equation} \label{exex'}
   \tilde{\rho}_2[\rho](x,y) = \frac12 \rho(x)\rho(y)- \frac12 \tau(x,y) 
   \quad \mbox{with} \quad
   \tau(x,y)=\sum_{\sigma,\sigma'}\left|\sum_{i=1}^N
   \phi_i(x,\sigma)\overline{\phi_i(y,\sigma')}\right|^2.
\end{equation}
Thus, just as in Example 2 the pair density naturally decomposes into a statistically independent term plus an exchange hole, but here the shape of the hole is no longer fixed but adapts itself to the density at hand. Expression \eqref{exex'} 
results in the interaction energy 
\begin{equation}
   \tilde{V}_{ee}[\rho] = \frac12 \int_{\R^6} \frac{\rho(x)\rho(y) - \tau(x,y)}{|x-y|} dx\, dy.
\end{equation}
The correction to \eqref{hartree} is known as exact exchange. 
Note that the resulting ground state energy \eqref{GSEapprox} is not quite the Hartree-Fock energy. This is because the orbitals are only determined via minimization of kinetic energy, rather than self-consistently accounting also for exchange. 
However, if one treats the orbitals $\Phi=(\phi_1,\cdots,\phi_N)$ as the basic variable, views the right hand side of \eqref{exex'} as an orbitals-to-pair-density map $\rho_2[\Phi](x,y)$, and substitutes into \eqref{Veeex} and \eqref{min-ks} one obtains precisely the Hartree-Fock energy. 
\\[3mm]
{\bf Example 4.} (Hybrid models) 
If we take some convex combination of \eqref{rho2-nonunif} and \eqref{exex'}, the interaction energy begins to resemble, up to certain further corrections, state of the art hybrid functionals such as B3LYP \cite{becke93,lee88,stephens94}, which are widely used in contemporary computations.
\\[3mm]
{\bf Example 5.} (strictly correlated electrons)
A more recent construction is the SCE (strictly correlated electrons) functional \cite{seidl99a,seidl99b,seidl07} 
\begin{equation} \label{VeeSCEttilde}
    \tilde{V}_{ee}[\rho] = {V}_{ee}^{SCE}[\rho] = \inf_{{T}_1,..,{T}_N}\int_{\R^3}\frac{\rho(z)}{N}
    \sum_{1\le i<j\le N} \frac{1}{|{T}_i(z)-{T}_j(z)|}dz,
\end{equation}
with the infimum taken over maps $T_1,...,T_N$ from $\R^3$ to $\R^3$ which satisfy $T_1(x)=x$ and which preserve $\rho$, that is to say
$$
 \int_A \rho = \int_{T_i(A)}\rho \mbox{ for all measurable sets }A\subset\R^3.
$$
This corresponds to the following density-to-pair-density map which we call $\rho_{2}^{SCE}[\rho]$: 
\begin{equation}\label{rho2-sce}
\tilde{\rho}_2[\rho](x,y) = \rho_2^{SCE}[\rho](x,y) = \frac{1}{2N}\sum_{i\neq j}\int_{\mathbb{R}^3}\rho(z)\delta(x-T_i(z))\delta(y-T_j(z))dz
\end{equation}
with the $T_i$ being minimizing maps. 
The physical meaning of the $T_i$ is that the position of one electron (at $x=T_1(x)$) fixes the positions of all the other $N-1$ electrons (at $T_i(x)$ with $2\leq i\leq N$).
Mathematically, the variational problem in \eqref{VeeSCEttilde} is a multi-marginal optimal transport problem. Minimizers are known to exist when $N=2$ \cite{cotar13a,buttazzo12} and $d=1$ \cite{colombo}. It is believed (and has been proved mathematically for $N=2$ \cite{cotar13a}) that ${V}_{ee}^{SCE}$ agrees with the lowest expectation of Coulomb repulsion energy with a given single-particle density $\rho$,
\begin{equation}\label{VeeSCE}
  \bar{V}_{ee}^{SCE}[\rho] = \inf_{\Psi\in\AN, \, \Psi\mapsto\rho}V_{ee}[\Psi].
\end{equation}
To derive \eqref{VeeSCEttilde} from \eqref{VeeSCE}, one notes that the infimum in \eqref{VeeSCE} is not attained in any reasonable wavefunction class such as \eqref{AN} or $\{\Psi\in L^2((\R^3\times\Z_2)^N)\, : \, \Psi \mbox{ antisymmetric, }||\Psi||_{L^2}=1\}$.
Therefore, one needs to augment the admissible $N$-body densities $\rho_N=\sum_{s_1,..,s_N\in\Z_2}|\Psi|^2$ in \eqref{VeeSCE} from integrable functions to probability measures, i.e. considers
\begin{equation}\label{Kant}
  \min_{\rho_N\mapsto\rho} \int_{\R^{3N}}\sum_{i<j}\frac{1}{|x_i-x_j|}d\rho_N,
\end{equation}
and makes the ansatz \cite{seidl99a,seidl99b}
\begin{equation} \label{Monge}
   \rho_N(x_1,\cdots,x_N) = \frac{1}{N!}\sum_{\mathcal{P}}\int_{\mathbb{R}^3}
   \frac{\rho(z)}{N}\prod_{i=1}^N\delta(x_i-T_{\mathcal{P}(i)}(z)) dz,
\end{equation}
where the sum runs over all permutations $\mathcal{P}$ of $\{1,..,N\}$.
Note that \eqref{rho2-sce} is obtained by integrating out all but two electron coordinates from this $\rho_N$.
The ansatz \eqref{Monge}, which reduces the high-dimensional problem \eqref{Kant} to a computationally feasible one, was later understood \cite{cotar13a,buttazzo12} as an instance of the mathematical belief that ``Kantorovich equals Monge'', i.e. that optimal Kantorovich transportation plans in are induced by Monge maps for well behaved marginal densities $\rho$ (see \cite{brenier,GMC,gangbo98} for pioneering results and \cite{villani09} for a comprehensive survey).   

Examples 1 to 4 are based on a non-interacting picture and treat many-body effects as corrections. Despite their great successes, these models exhibit known failures for strongly interacting systems \cite{cohen12}.
By comparison, Example 5 takes the strongly interacting limit and has been proved to be good at
simulating some strongly correlated model systems
(e.g. \cite{malet12,mendl13}), but severely underestimates the true ground state energy in standard regimes (see e.g. the dissociation curve of the hydrogen dimer calculated in \cite{ChenEtAl}). It is therefore of great interest to enquire as to the structure and behaviour of the true pair densities $\rho_2[\rho]$. 
\section{Density scaling, adiabatic connection, formal asymptotics}\label{sec-hk}
\setcounter{equation}{0}\setcounter{figure}{0}
In order to naturally access pair densities in different correlation regimes without changing the ``shape'' of the one-body density, we will from now on look at one-parameter families of one-body densities obtained by rescaling a fixed reference density $\rho \, : \, \R^d\to\R$ (see eq. \eqref{scale_intro}). The associated pair densities do {\it not} just change by a rescaling (see the Introduction). This reflects the physical phenomenon that electron correlation in dilute systems ($\alpha<<1$) is completely different from electron correlation in high-density systems ($\alpha>>1$). The governing variational principle for the resulting constrained-search wavefunction in \eqref{UPD} was found by Levy and Perdew \cite{PerdewLevy}. As a straightforward corollary of their analysis we obtain the behaviour of the density-to-pair-density map under density scaling:
\begin{proposition} (Density scaling) \label{propo-Falpha-1} 
Let $\alpha>0$ and let $\rho$ be any single-particle density on $\R^d$, i.e. any function belonging to the class $\RN$. Then the diagram \eqref{pdscale_intro} commutes. In other words, if $\rho_{2,\alpha}[\rho]$ denotes the density-to-pair-density map along the adiabatic connection (left arrow in the diagram), that is to say
\begin{equation} \label{SPD}
   \rho_{2,\alpha}[\rho] := \{\rho_2^{\Psi}\, | \, \Psi \mbox{ is a minimizer of }\alpha T+V_{ee} \mbox{ on } \AN \mbox{ s/to }\Psi\mapsto\rho\},
   \footnote{``s/to" means ``subject to" throughout this paper.}
\end{equation}
and $\rho_2[\rho]$ is the original map \eqref{UPD}, then 
\begin{equation}\label{diagram}
   D_{\alpha^{-1}}\rho_2[D_\alpha\rho] = \rho_{2,\alpha}[\rho].
\end{equation}
\end{proposition}
\begin{proof} For convenience of the reader we include the simple proof. For any $\alpha>0$ and any $\Psi\in\AN,\,\Psi\mapsto\rho$, we can rescale $\Psi$ by
$$
\Psi_{\alpha}(x_1,\cdots,x_N)=\alpha^{dN/2}\Psi(\alpha x_1,\cdots,\alpha x_N).
$$
We have that $\Psi_{\alpha}$ belongs to $\AN$ and has one-body density $D_\alpha\rho$. 
Moreover 
\begin{eqnarray*}
T[\Psi_{\alpha}]=\alpha^2T[\Psi]
\quad\quad{\rm and}\quad\quad
V_{ee}[\Psi_{\alpha}]=\alpha V_{ee}[\Psi].
\end{eqnarray*}
It follows that $\Psi_\alpha$ is a minimizer of $T+V_{ee}$ subject to $\Psi_\alpha\mapsto D_\alpha\rho$ if and only if $\Psi$ is a minimizer of $\alpha T + V_{ee}$ subject to $\Psi\mapsto\rho$. By definition, the pair densities of the minimizing $\Psi_\alpha$'s yield the set $\rho_2[D_\alpha\rho]$, whereas the pair densities of the associated $\Psi$'s give the set $\rho_{2,\alpha}[\rho]$. 
\end{proof}
From now on, instead of considering the scaled densities \eqref{scale_intro} and applying the original density-to-pair-density map, it is more convenient for us to fix a reference single-particle density and vary the coupling constant $\alpha$ in the constrained-search problem in \eqref{SPD} (i.e., in the adiabatic connection) to investigate systems in different correlation regimes. 

We note that definition \eqref{SPD} stays unchanged under multiplying $\alpha T + V_{ee}$ by a positive constant, so one might as well use $T+\alpha^{-1}V_{ee}$. In particular, one has a well-defined density-to-pair-density map at $\alpha=\infty$:
$$
   \rho_{2,\infty}[\rho] = \{\rho_2^\Psi\, : \, \Psi\in\AN, \, \Psi \mbox{ is a minimizer of }T\mbox{ on }\AN \mbox{ s/to }\Psi\mapsto\rho\}. 
$$
It is considered well-established in the physics literature (see e.g. \cite{Seidl2000}) that the minimizing wavefunction in \eqref{SPD} has the following asymptotic behaviour: 
\begin{equation} \label{asylarge}
   \Psi \approx \mbox{Slater determinant of the KS orbitals 
    from Ex. 3, Section \ref{sec:approx}} \;\;\; (\alpha >> 1)
\end{equation}
and 
\begin{equation} \label{asysmall}
   \sum_{s_1,..,s_N\in\Z_2}|\Psi|^2 \approx \mbox{$N$-point density of the SCE state, eq. \eqref{Monge}} \;\;\; (\alpha << 1).
\end{equation}
Taking pair densities leads to 
\begin{equation} \label{pdasy}
   \rho_{2,\alpha} \approx \begin{cases} \mbox{\eqref{exex'}}, & \alpha>>1, \\
                                         \mbox{\eqref{rho2-sce}}, & \alpha << 1. \end{cases}
\end{equation}
Complete mathematical proofs are not available for general $\rho$.
It is not clear in which sense to measure convergence, nor what happens if the ground state is degenerate.
In fact, even much more basic things such as existence of optimal maps or continuity of the HK functional have not been proved.  
The rigorous analysis of 1D examples in Section \ref{sec:asy} shows that things are not quite as simple as one might intuitively expect. For instance, in case of orbital degeneracies the assertion \eqref{asylarge} can be true for {\it some} choices of minimizing KS orbitals but not for others.   

At least for $N=2$ or in the case of bosons we can offer a general result. For bosons, the set $\AN$ of antisymmetric wavefunctions has to be replaced by 
\begin{equation}\label{BN}
   {\cal B}_N = \{\Psi\in L^2(\R^{d\cdot N}), \, \nabla\Psi\in L^2, \, \Psi \mbox{ symmetric}, \, ||\Psi||_{L^2(\R^{d\cdot N})}=1\}.    
\end{equation}
\begin{proposition}\label{propo-Falpha-3}
Let $\rho$ be any single-particle density of an $N$-particle system, i.e. $\rho\in\RN$.
If $N=2$, or if the particles are bosons, then the independent pair density
\begin{equation}\label{pdbos}
    \frac12 \Bigl(1-\frac1{N}\Bigr)\rho(x)\rho(y)
\end{equation}
belongs to the set $\rho_{2,\infty}[\rho]$.
\end{proposition}
\begin{proof} We claim that the product wave function
$
             \tilde{\Psi}(x_1,..,x_N)=\prod_{i=1}^N\sqrt{\rho(x_i)/N},
$
which has pair density \eqref{pdbos}, is a minimizer of $T$ on ${\cal B}_N$ subject to the constraint $\Psi\mapsto\rho$. 
To see this, consider a general $\Psi\in{\cal B}_N$ with $\Psi\mapsto\rho$, and estimate
\begin{eqnarray*}
T[\Psi] &=& \frac{1}{2}\sum_{i=1}^N \int_{\mathbb{R}^{dN}}|\nabla_{x_i}\Psi|^2dx_1\dots dx_N 
= \frac{N}{2}\int_{\R^{dN}} |\nabla_{x_1}\Psi|^2dx_1\dots dx_N \\[1ex]
&\geq & \frac{N}{2}\int_{{\R}^d}
\frac{\left| \int_{\mathbb{R}^{d(N-1)}}{\rm Re}(\overline{\Psi}\nabla_{x_1}\Psi)dx_2\dots dx_N\right|^2}
{\int_{\mathbb{R}^{d(N-1)}}|\Psi|^2 dx_2\dots dx_N} dx_1 \\[1ex]
&=& \frac{N}{8}\int_{\mathbb{R}^d}
\frac{\left|\nabla_{x_1}\int_{\mathbb{R}^{d(N-1)}}|\Psi|^2dx_2 \dots dx_N\right|^2}
{\rho(x_i)} dx_1 = \frac{1}{8}\int_{\mathbb{R}^d}\frac{|\nabla\rho(x)|^2}{\rho(x)}dx.
\end{eqnarray*}
On the other hand, by an elementary calculation, $T[\tilde{\Psi}]$ is equal to the expression in the last line. For fermions with $N=2$, analogous arguments show that 
the Slater determinant with orbitals $\sqrt{\rho(x)/N}\uparrow(s)$, $\sqrt{\rho(x)/N}\downarrow(s)$ is a minimizer.
\end{proof}
%
%

\section{Numerical investigations of the pair densities}\label{sec-numerical}
\setcounter{equation}{0}\setcounter{figure}{0}

We now turn to the intermediate regime where $\alpha$ lies somewhere inbetween zero and infinity, and investigate numerically how the crossover between the limit behaviour 
\eqref{asysmall} and \eqref{asylarge} occurs. To this end we compute, for simple reference densities $\rho$, the whole one-dimensional family of pair densities $\rho_{2,\alpha}[\rho]$ 
($\alpha\in(0,\infty)$) along the adiabatic connection. Recall that each $\rho_{2,\alpha}$ arises, up to a re-scaling, as a true pair density (see \eqref{diagram}). 

Due to the nontrivial (infinite-dimensional, nonlinear) constraint $\Psi\mapsto\rho$ and the need to resolve $N$-electron wavefunctions, we limit ourselves here for simplicity to one-dimensional reference densities $\rho$ and particle numbers $N=2,3,4$. We hope that our results are nevertheless of some physical and chemical interest. 

Note that the one dimensional Coulomb repulsion can not be described by $1/|x|$ since the latter function is not integrable near $0$. We therefore 
use an effective potential $c(|x|)$
which is obtained by integrating the Coulomb repulsion in $\mathbb{R}^3$ in a thin wire over the lateral degrees of freedom \cite{bednarek03}. Explicitly,
\begin{eqnarray*}
c(r) = \frac{\sqrt{\pi}}{2b}{\rm exp}\left(\frac{r^2}{4b^2}\right){\rm erfc}\left(\frac{r}{2b}\right),
\end{eqnarray*}
where $b$ is a constant and ${\rm erfc}$ is the complementary error function.
We set $b=0.1$ in our simulations (see Figure \ref{fig-br}).

Let $\Omega=[-L,L]$ (with $L=5.0$ in the simulations) and let $N$ be the particle number.
We consider two typical systems on $\Omega$ (see Figure \ref{fig-rho}):
a homogeneous density with periodic boundary condition
\begin{equation}\label{rho-hom}
\rho(x)\equiv\frac{N}{2L};
\end{equation}
and a smoothly varying density with zero Dirichlet boundary condition
\begin{eqnarray}\label{rho-inhom}
\rho(x)=\frac{N}{2L}\left(1+\cos(\frac{\pi}{L}x)\right).
\end{eqnarray}
Both of these two single-particle densities belong to space \eqref{RN} (with $\R^d$ replaced by $\Omega$).

\begin{figure}[ht]
	\begin{minipage}[t]{0.5\linewidth}
		\centering
		\includegraphics[height=5.0cm]{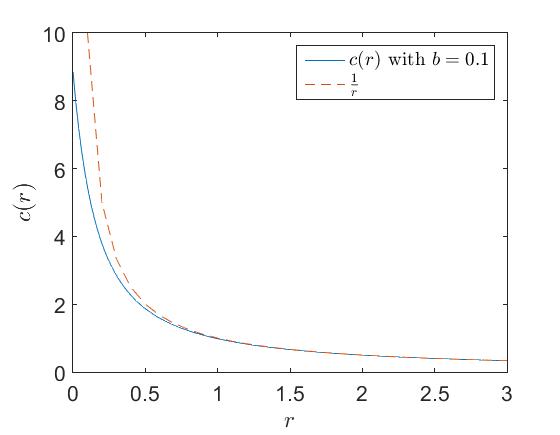}
		\caption{The one dimensional effective Coulomb potential.}
		\label{fig-br}
	\end{minipage}
	\hskip 0.2cm
	\begin{minipage}[t]{0.5\linewidth}
		\centering
		\includegraphics[height=5.0cm]{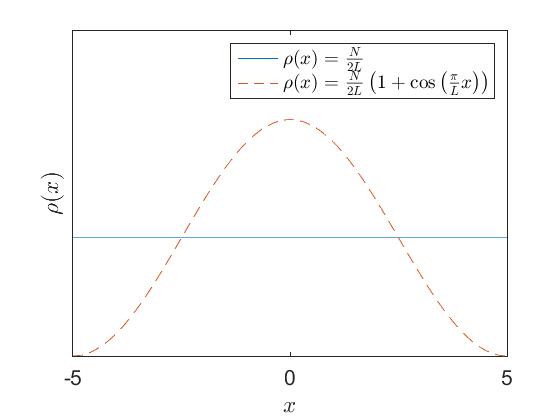}
		\caption{The homogeneous and inhomogeneous electron density (6.1) and (6.2).}
		\label{fig-rho}
	\end{minipage}
\end{figure}

For later purposes, we calculate the optimal transport maps by using the formulae in \cite{seidl07} (which were recently justified rigorously in \cite{cotar13a} for $N=2$ and in \cite{colombo} for general $N$, and are described in Theorem \ref{ThmA} below) 
and present them in Figure \ref{fig-rho_per} and \ref{fig-rho_cos} for the two systems with 2, 3, and 4 particles. 

\begin{figure}[!htb]
\centering
\includegraphics[width=4.5cm]{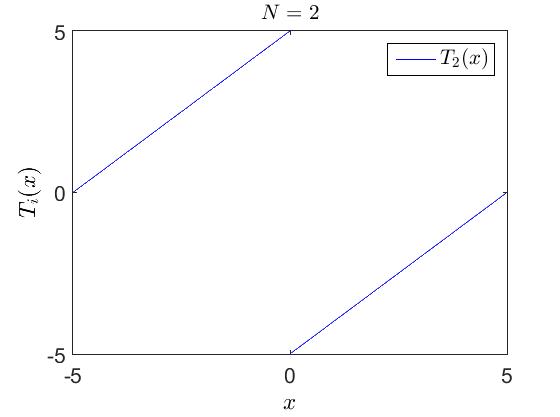}
\includegraphics[width=4.5cm]{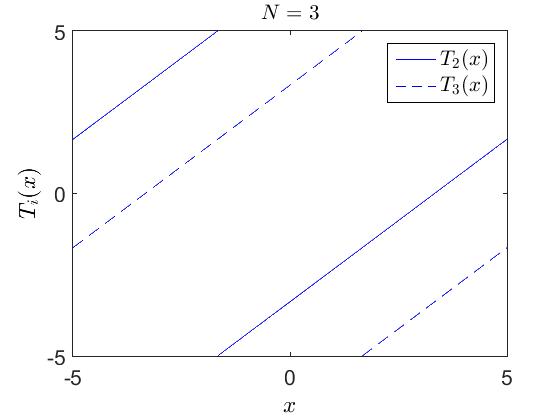}
\includegraphics[width=4.5cm]{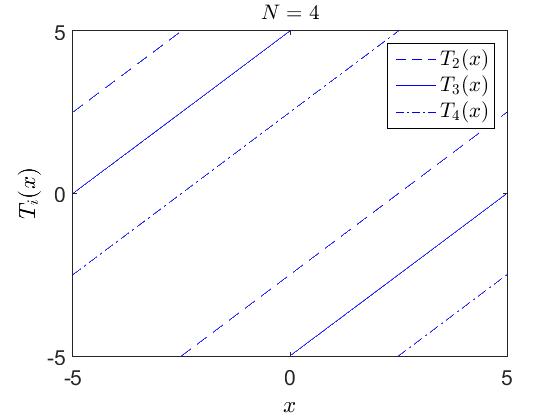}
\caption{The optimal transport maps of the 
	density $\rho(x)\equiv\frac{N}{2L}$.}
\label{fig-rho_per}
\end{figure}

\begin{figure}[!htb]
\centering
\includegraphics[width=4.5cm]{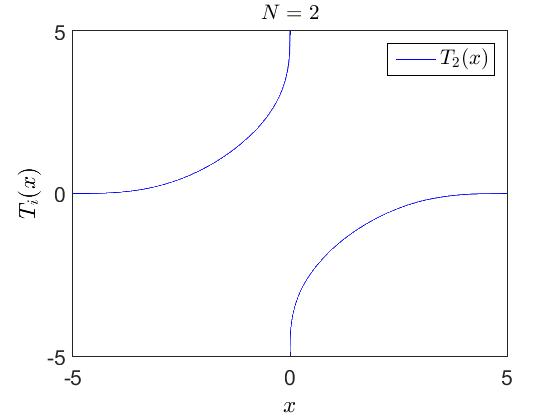}
\includegraphics[width=4.5cm]{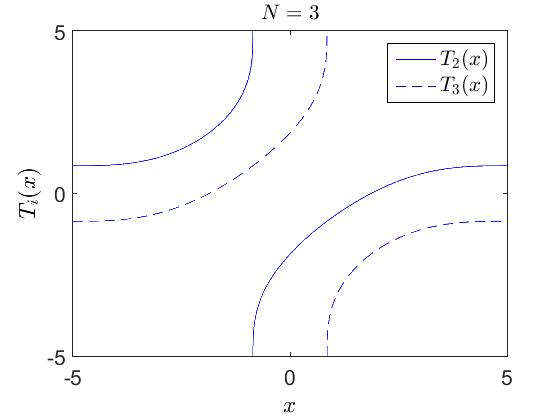}
\includegraphics[width=4.5cm]{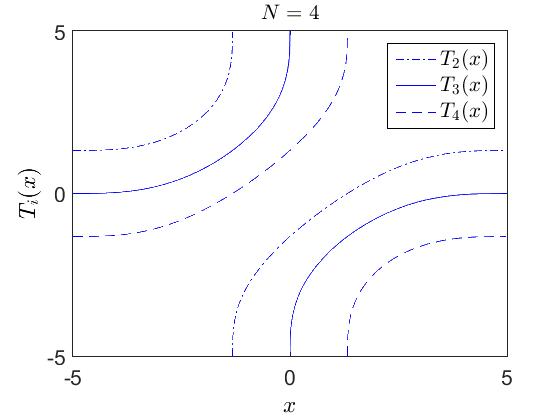}
\caption{The optimal transport maps of the 
	density $\rho(x)=\frac{N}{2L}(1+\cos(\frac{\pi}{L}x))$.}
\label{fig-rho_cos}
\end{figure}

Moreover, in one dimension it is known rigorously \cite{colombo} that the maps $T_1,..,T_N$ are cyclic, that is to say
$$
  \underbrace{T_2\circ ... \circ T_2}_{N\mbox{ times}} = id, \;\;\;
  \underbrace{T_2\circ ... \circ T_2}_{j-1\mbox{ times}} = T_j \; (j=2,..,N).
$$
For related insights see \cite{Ghoussoub}.
This allows to simplify formula \eqref{rho2-sce} for the $\alpha=0$ pair density $\rho_{2}^{SCE}$. Namely, a change of variables shows that in this case the sum over $j$ in \eqref{rho2-sce} is independent of $i$. This together with the fact that the normalized line element (one-dimensional Hausdorff measure) $ds$ on the one-dimensional curve {\it graph} $\,T_j$=$\{(x,T_j(x))\, : \, x\in\R\}$ is given by 
$$
    ds = \sqrt{1+T_j'(x)^2}dx ,
$$
which yields the expression
\begin{equation} \label{rho2geo}
   \rho_{2}^{SCE}[\rho](x,y) = \frac12 \sum_{j=2}^N \frac{\rho(x)}{\sqrt{1+T_j'(x)^2}} ds \Big|_{y=T_j(x)}.
\end{equation}
This remarkable formula shows that the maps $T_i$, and hence the full $N$-body SCE density, can be {\it explicitly read off} from the SCE pair density!

To obtain the true pair densities of our two typical systems
for finite coupling constant $\alpha$, we need to simulate the constrained-search 
problem in \eqref{SPD}. In our case this problem is 
given, for a one-dimensional single-particle density $\rho_0$, by
\begin{multline}\label{min_F_alpha}
 \mbox{Minimize }\sum_{\sigma_1,..,\sigma_N\in\Z_2}\int_{\Omega^{N}}
\left(\frac{\alpha}{2}\sum_{i=1}^N \left|\frac{\partial\Psi}{\partial x_i}\right|^2
+ \sum_{1\leq i<j \leq N} |\Psi|^2 c(|x_i-x_j|)\right) dx_1\cdots dx_N \\
 \mbox{ s/to }\Psi\in\AN,\,\Psi\mapsto\rho_0. \quad
\end{multline}
Here for the inhomogeneous density \eqref{rho-inhom} $\AN$ is the standard wavefunction class \eqref{AN} with $\R^d$ replaced by $\Omega=[-L,L]$, 
\begin{equation}\label{ANDir}
  \AN = \Bigl\{\Psi\, : \, \Psi,\,\nabla\Psi\in L^2(([L,L]\times\Z_2)^N), \, \Psi\mbox{ antisymmetric}, \, ||\Psi||_{L^2}=1\Bigr\}.
\end{equation}
Since $\rho$ is zero at $\pm L$, the constraint $\Psi\mapsto\rho$ automatically implies Dirichlet zero boundary conditions $\Psi\!\Big|_{x_i=\pm L}=0$. For the homogeneous density \eqref{rho-hom}, we use periodic wavefunctions
\begin{eqnarray} 
  \AN &=& \Bigl\{ \Psi\, : \, \Psi,\,\nabla\Psi\in L^2(([-L,L]\times\Z_2)^N), \nonumber \\
  & & \Psi \mbox{ antisymmetric}, \, \Psi\Bigl|_{x_i=L}=\Psi\Bigr|_{x_i=-L}\mbox{ for all }i, 
       \, ||\Psi||_{L^2}=1 \Bigr\} . \label{ANper}
\end{eqnarray}  
Recall that the constraint $\Psi\mapsto\rho_0$ means that integrating $|\Psi|^2$ over all but one electron positions and summing over all spins gives the single-particle density $\rho_0$.
By the symmetry of $|\Psi|^2$, one can leave any of the electron coordinate not to be integrated.
Therefore, the associated Lagrange function of \eqref{min_F_alpha} is, abbreviating $z_i=(x_i,\sigma_i)\in\Omega\times\Z_2$ and $\sum_{\sigma_i}\int_{\Omega}dx_i = \int_{\Omega\times\Z_2} dz_i$, 
\begin{multline*}
\mathcal{L}(\Psi,\lambda_1,\cdots,\lambda_N) = \int_{(\Omega\times\Z_2)^{N}}\left(\frac{\alpha}{2}
\sum_{i=1}^N \left|\frac{\partial\Psi}{\partial x_i}\right|^2
+ \sum_{1\leq i<j \leq N} |\Psi|^2 c(|x_i-x_j|)\right) dz_1\cdots dz_N \\[1ex]
+ \int_{\Omega}\lambda_1(x)\left(\rho_0(x) - N\sum_{\sigma\in\Z_2}\int_{(\Omega\times\Z_2)^{N-1}}
|\Psi(x,\sigma, z_2,\cdots,z_N)|^2 dz_2\cdots dz_N\right) dx \\[1ex]
+ \cdots + \int_{\Omega}\lambda_N(x)\left(\rho_0(x) - N\sum_{\sigma\in\Z_2}\int_{(\Omega\times\Z_2)^{N-1}}
|\Psi(z_1,\cdots,z_{N-1},x,\sigma)|^2 dz_1\cdots dz_{N-1}\right) dx
\end{multline*}
with the Lagrange multipliers $\lambda_1(x),\cdots,\lambda_N(x)$.
By the symmetry of $|\Psi|^2$, we have $\lambda_1(x)=\cdots=\lambda_N(x) :=\lambda(x)$.
Therefore minimizers of \eqref{min_F_alpha} satisfy the following Euler-Lagrange equation
\begin{eqnarray}\label{eq-euler-lag}
\left\{\begin{array}{l} \displaystyle
\left(-\frac{\alpha}{2}\Delta + \sum_{1\leq i<j\leq N}c(|x_i-x_j|)\right) \Psi(x_1,\sigma_1,\cdots,x_N,\sigma_N) 
\\ \qquad\qquad\qquad\qquad\qquad\qquad\qquad\qquad  \displaystyle
= \left(\sum_{i=1}^N\lambda(x_i)\right) \Psi(x_1,\sigma_1,\cdots,x_N,\sigma_N) \\ 
\Psi\in\AN,~\Psi\mapsto\rho_0
\end{array}\right..\quad
\end{eqnarray}

Formally, the Lagrange multiplier $\lambda(x)$ equals the functional derivative of the Hohenberg-Kohn functional \eqref{min_F_alpha} with respect to electron density, and $-\lambda(x)$ equals the external potential $v_0$ for which $\rho_0$ is the ground state of the system,
$$
\lambda = \left.\frac{\partial F_{\alpha}[\rho]}{\partial\rho}\right|_{\rho=\rho_0} = -v_0
$$
with $F_{\alpha}[\rho] := \min_{\Psi\in\AN,\Psi\mapsto\rho} \left\{ \alpha T[\Psi] + V_{ee}[\Psi] \right\}$.
Therefore, by using the Euler-Lagrange equation \eqref{eq-euler-lag} we implicitly require that the density $\rho_0$ can be generated by some external potential.\footnote{We thank Eric Cances for this remark.} 
This is called ``$v$-representability" \cite{martin05}, and the conditions for such densities are not known in general. In particular, we do not know rigorously whether the single-particle densities given by \eqref{rho-hom} and \eqref{rho-inhom} are $v$-representable.
Nevertheless, after numerical discretization it can easily be shown that
the Lagrange multiplier $\lambda(x)$ for the ensuing finite dimensional problem exists. Moreover our numerical Lagrange multipliers stayed stable under refining the mesh, suggesting that $v$-representability holds. Establishing this rigorously is an interesting open problem.

Next we describe our algorithm for solving \eqref{eq-euler-lag}. We drop the spin variables for simplicity; extension to the spin-dependent case is straightforward. Equation \eqref{eq-euler-lag} looks like an eigenvalue problem, but the ``eigenvalue" $\sum_{i=1}^N\lambda(x_i)$
depends on a function on $\Omega$, and moreover, the ``eigenfunction" $\Psi$ has to satisfy some
nonlinear marginal constraints. Due to these difficulties, there is no simple way for us to solve this problem directly.
If we look at the equation \eqref{eq-euler-lag} the other way around by assuming that $\lambda(x) = \tilde{\lambda}(x)$
with some given function $\tilde{\lambda}$, then the problem is reduced to the following generalized eigenvalue problem:
Find $\mu\in\mathbb{R}$ and $0\neq\tilde{\Psi}\in H^1(\Omega^N)$, such that
\begin{eqnarray}\label{general-eigen}
\left(-\frac{\alpha}{2}\Delta + \sum_{1\leq i<j\leq N}c(|x_i-x_j|)\right) \tilde{\Psi}(x_1,\cdots,x_N)
= \mu\left(\sum_{i=1}^N\tilde{\lambda}(x_i)\right) \tilde{\Psi}(x_1,\cdots,x_N)
\end{eqnarray}
with $\|\tilde{\Psi}\|_{L^2}=1$ and $\mu$ being the lowest eigenvalue.
We thus obtain an eigenfunction $\tilde{\Psi}$ with corresponding single-particle density
\begin{eqnarray*}
\tilde{\rho}(x)=N\int_{\Omega^{N-1}}\left|\tilde{\Psi}(x,x_2,\cdots,x_N)\right|^2dx_2\cdots dx_N.
\end{eqnarray*}
We denote the above process (from $\tilde{\lambda}$ to $\tilde{\rho}$) by $\mathcal{F}$,
that is, $\tilde{\rho}=\mathcal{F}(\tilde{\lambda})$.
We have that \eqref{eq-euler-lag} is equivalent to the nonlinear problem
\begin{eqnarray}\label{rho_lambda}
\rho_0 = \mathcal{F}(\lambda).
\end{eqnarray}
We resort to the following Newton algorithm for solving this nonlinear problem.

\begin{algorithm}\label{algorithm-newton}
{\bf Newton algorithm for solving \eqref{rho_lambda}}
\begin{enumerate}
\item
Fix $\varepsilon>0$. Let $k=1$ and initialize $\lambda_1\neq 0$.
\item
Solve \eqref{general-eigen} with $\tilde{\lambda}=\lambda_k$ to obtain $(\mu_k,\Psi_k)$.
\item
Let $\lambda'_k=\mu_k\lambda_k$ and
$\rho_k(x)=N\int_{\Omega^{N-1}}|\Psi_k(x,x_2,\cdots,x_N)|^2 dx_2\cdots dx_N$.
\item
If $\|\rho_0-\rho_k\|<\varepsilon$, get $\Psi=\Psi_k$ and goto 5; else, let
\begin{eqnarray}\label{iteration}
\lambda_{k+1}=\lambda'_k + \left(\left.\frac{\partial \mathcal{F}}{\partial\lambda}
\right|_{\lambda=\lambda'_k}\right)^{-1}(\rho_0-\rho_k).
\end{eqnarray}
Take $k=k+1$ and goto 2.
\item
Calculate the pair density
$\rho_2(x,y)=\binom{N}{2}\int_{\Omega^{N-2}}|\Psi_k(x,y,x_3,\cdots,x_N)|^2 dx_3\cdots dx_N$.
\end{enumerate}
\end{algorithm}

Note that the operator
$\left(\left.\frac{\partial \mathcal{F}}{\partial\lambda}\right|_{\lambda=\lambda_k}\right)^{-1}$
in \eqref{iteration} can not be obtained explicitly, an approximation for it has to be made.
We abbreviate $\Lambda_k(x_1,\cdots,x_N)=\sum_{i=1}^N\lambda_k(x_i)$ and obtain by the chain rule that
\begin{eqnarray}\label{chain-rule}
\left.\frac{\partial \mathcal{F}}{\partial\lambda}\right|_{\lambda=\lambda_k}
= \frac{\partial\rho_k}{\partial\Psi_k} \cdot \frac{\partial\Psi_k}{\partial\Lambda_k}
\cdot \frac{\partial\Lambda_k}{\partial\lambda_k}.
\end{eqnarray}
The first and third factors on the right-hand side of \eqref{chain-rule} can be obtained explicitly.
To calculate the second term, we observe that
\begin{eqnarray}\label{eq-PsiLambda}
\mathcal{H}\Psi_k=\Lambda_k\Psi_k \quad{\rm with}~~
\mathcal{H}=-\frac{\alpha}{2}\Delta + \sum_{1\leq i<j\leq N}c(|x_i-x_j|).
\end{eqnarray}
By differentiating \eqref{eq-PsiLambda} with respect to $\Lambda_k$ and ignoring the $\Lambda_k$-dependence of $\Psi_k$ on the right-hand side, we can obtain the approximation
$$
\frac{\partial\Psi_k}{\partial\Lambda_k} \approx \mathcal{H}^{-1}\Psi_k.
$$
In our numerical experiments, the single-particle densities $\rho_k$ generated by Algorithm \ref{algorithm-newton}
always converged to $\rho_0$ steadily (see Figure \ref{fig-newton} in Section \ref{sec-fermion} as an example).

We can solve any discretization of \eqref{min_F_alpha} numerically by using Algorithm \ref{algorithm-newton}, and further obtain the true pair densities for different coupling constants.
We perform all our following computations in double precision arithmetic on a PC with 16GB RAM using Matlab.

\subsection{Bosons}\label{sec-boson}

To elucidate pure correlation effects undiluted by exchange, we first neglect the spin variables and assume that $\Psi$ is symmetric,
that is, we assume that the particles under consideration are bosons. 

Let $\mathcal{T}$ be a partition of $\Omega=[-L,L]$ ($L=5$) with equally spaced nodes $a^1<a^2<\cdots<a^m$.
Denote by $\chi_j(x)$ the piecewise linear function with value 1 at node $a^j$ and 0 otherwise.
Then the functions
\begin{eqnarray*}
\{\chi_j(x):~j=1,\cdots,m\}
\end{eqnarray*}
form a linear finite element basis set on $\Omega$,
which gives a discretization for the single-particle space.
Denote the finite dimensional space span$\{\chi_j:~j=1,\cdots,m\}$ by $V_m$.

Since the wavefunction $\Psi$ in \eqref{min_F_alpha} is a function on $\Omega^N$, we shall generate a basis set in $N$-particle space by taking tensor products of the $\{\chi_j\}$:
\begin{eqnarray}\label{basis-tensor}
\psi_{\bf j}({\bf x}) = \prod_{l=1}^N \chi_{j_l}(x_l)
\quad{\rm with}~~{\bf j}=(j_1,\cdots,j_N)\in\{1,\cdots,m\}^N.
\end{eqnarray}
Note that the number of degrees of freedom for this basis set is $m^N$.
We denote by $\mathcal{B}_m^N$ the $N$-boson space spanned by the basis functions $\{\psi_{\bf j}\}$.

With the above discretization, we have the following variational formulation of \eqref{eq-euler-lag}:
Find $\lambda\in V_m$ and $\Psi\in\mathcal{B}_m^N$ such that
\begin{eqnarray}\label{eq-euler-boson-dis}
\left\{\begin{array}{l} \displaystyle
\frac{\alpha}{2}(\nabla\Psi,\nabla v) + \sum_{1\leq i<j\leq N}(c(|x_i-x_j|)\Psi,v)
=\sum_{i=1}^N(\lambda(x_i)\Psi,v)\quad \forall~v\in\mathcal{B}_m^N \\ [1ex]
\displaystyle \rho(x)=N\int_{\Omega^{N-1}}|\Psi(x,x_2,\cdots,x_N)|^2 dx_2\cdots dx_N
\quad{\rm with}~ x=a^1,a^2,\cdots,a^m
\end{array}\right..\quad
\end{eqnarray}
The second line of \eqref{eq-euler-boson-dis} is a discretization of the
marginal constraint, which is only imposed on the nodes of $\mathcal{T}$.
Within this discretization, $\rho-\rho_k$ in Algorithm \ref{algorithm-newton} is
calculated as a vector $\{\rho(a^j)-\rho_k(a^j)\}_{j=1}^m$ on the nodes.

For the homogeneous density \eqref{rho-hom} and the inhomogeneous density \eqref{rho-inhom} with $N=2$, 3, and 4 particles,
we compute their pair densities by 
using Algorithm \ref{algorithm-newton} with $\varepsilon=10^{-4}$ and $m=40$ for $N=2,3$, $m=32$ for $N=4$.
The results for different values of $\alpha$ are presented in
Figure \ref{fig-bos-periodic} and \ref{fig-bos-convex}.
When $\alpha=0$, the electrons are strictly correlated to each other: the position of
one electron fixes all positions of the other electrons, and the pair densities are given by
\eqref{rho2geo}, with support ${\rm supp}(\rho_{2}^{SCE}[\rho]) = \{(x,T_i(x))\, : \, x\in\Omega,\, i=2,..,N\}$. To visualize this limiting pair density, we plot, above each curve 
$\{(x,T_i(x))\, : \, x\in\Omega\}$, the prefactor  $\rho(x)/\sqrt{1+(T'_i)^2}$ of the normalized line element $ds$ along the curve. 

\begin{figure}[htbp]
\centering

\begin{picture}(100,580)
\put(-40,540){\makebox(2,2){$N = 2$}}
\put(80,540){\makebox(2,2){$N = 3$}}
\put(200,540){\makebox(2,2){$N = 4$}}

\put(-140,480){\makebox(2,2){$\alpha = 0$}}
\put(-100,430){\includegraphics[width=4.5cm]{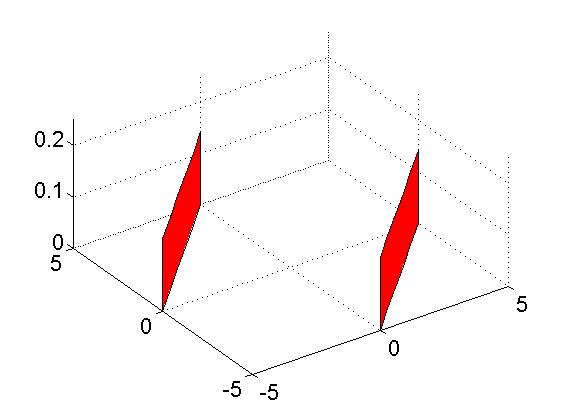}}
\put(20,430){\includegraphics[width=4.5cm]{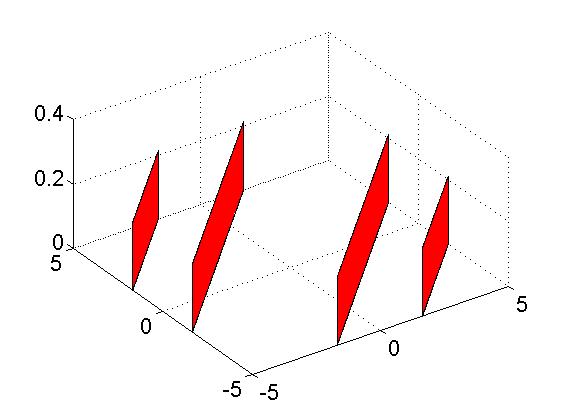}}
\put(140,430){\includegraphics[width=4.5cm]{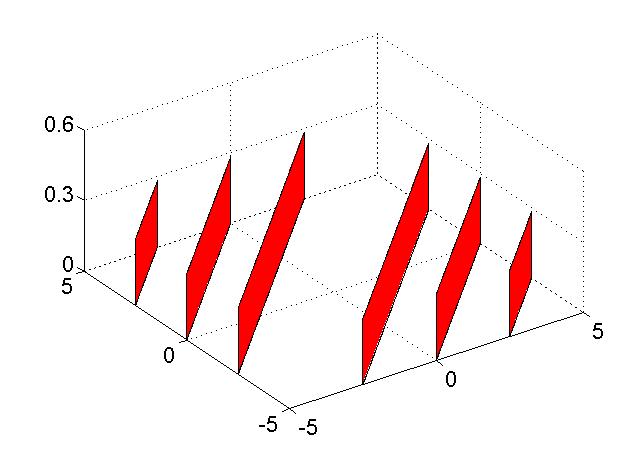}}

\put(-140,380){\makebox(2,2){$\alpha = 0.1$}}
\put(-100,330){\includegraphics[width=4.5cm]{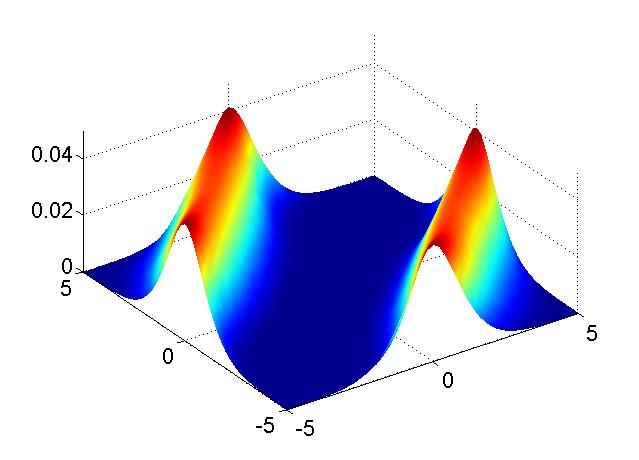}}
\put(20,330){\includegraphics[width=4.5cm]{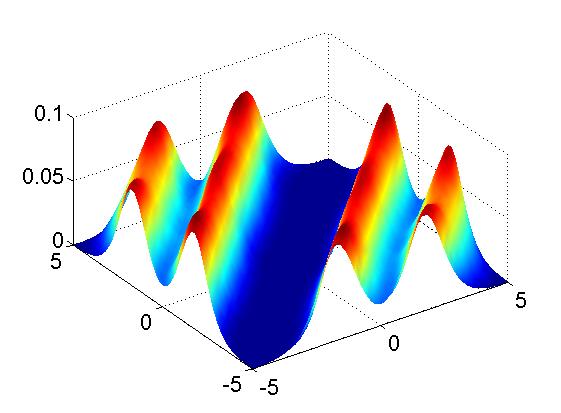}}
\put(140,330){\includegraphics[width=4.5cm]{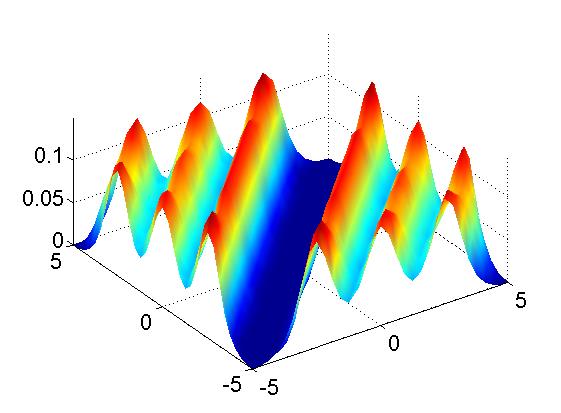}}

\put(-140,280){\makebox(2,2){$\alpha = 1$}}
\put(-100,230){\includegraphics[width=4.5cm]{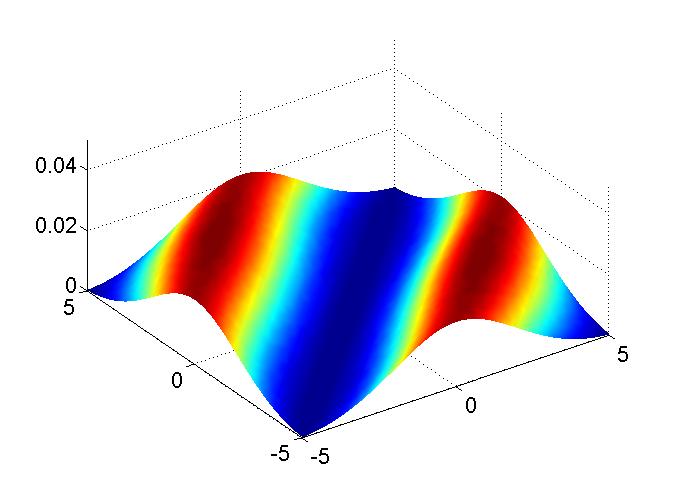}}
\put(20,230){\includegraphics[width=4.5cm]{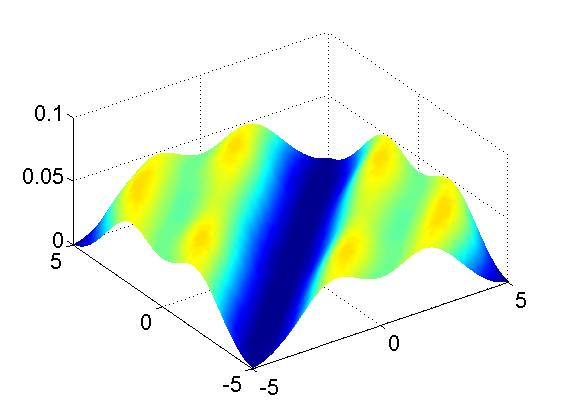}}
\put(140,230){\includegraphics[width=4.5cm]{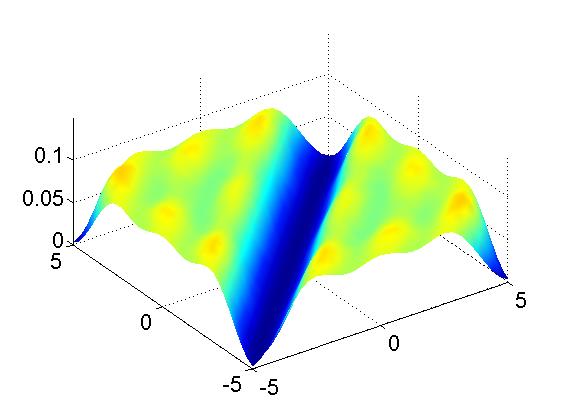}}

\put(-140,180){\makebox(2,2){$\alpha = 10$}}
\put(-100,130){\includegraphics[width=4.5cm]{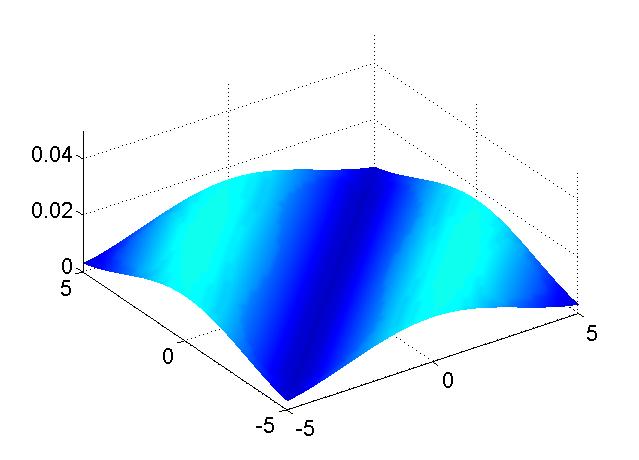}}
\put(20,130){\includegraphics[width=4.5cm]{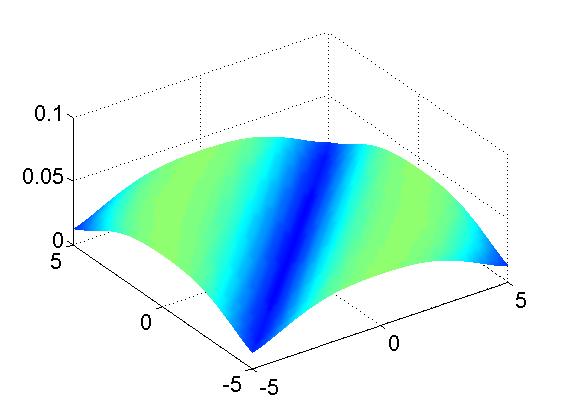}}
\put(140,130){\includegraphics[width=4.5cm]{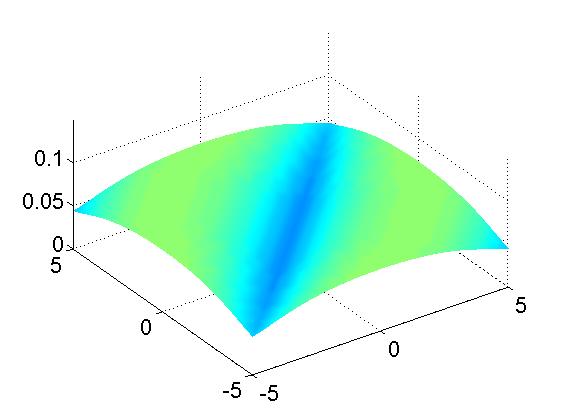}}

\put(-140,80){\makebox(2,2){$\alpha = 100$}}
\put(-100,30){\includegraphics[width=4.5cm]{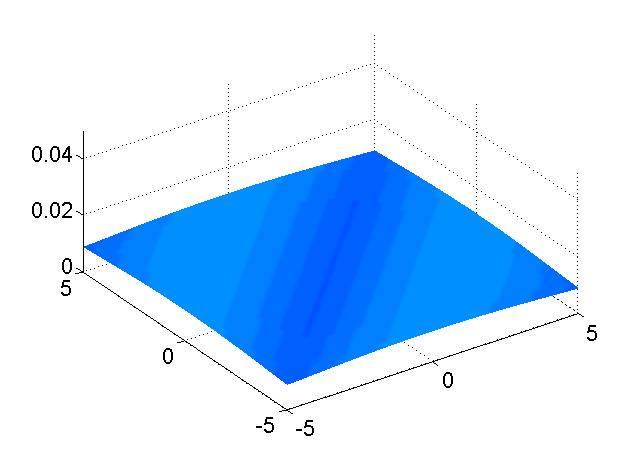}}
\put(20,30){\includegraphics[width=4.5cm]{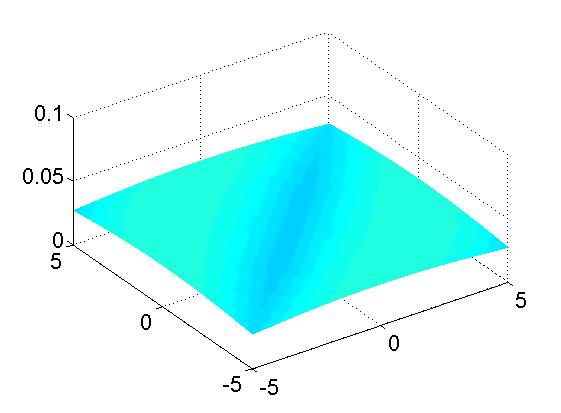}}
\put(140,30){\includegraphics[width=4.5cm]{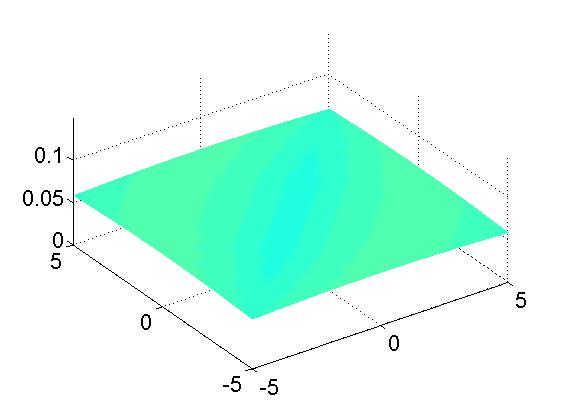}}

\end{picture}
\caption{(bosons) Pair densities with $\rho(x)\equiv\frac{N}{2L}\;$ (eqs.~\eqref{SPD},~\eqref{ANper}, but with spinless symmetric $\Psi$). Top row: SCE/optimal transport (based on exact results \cite{seidl99a}).}
\label{fig-bos-periodic}
\end{figure}
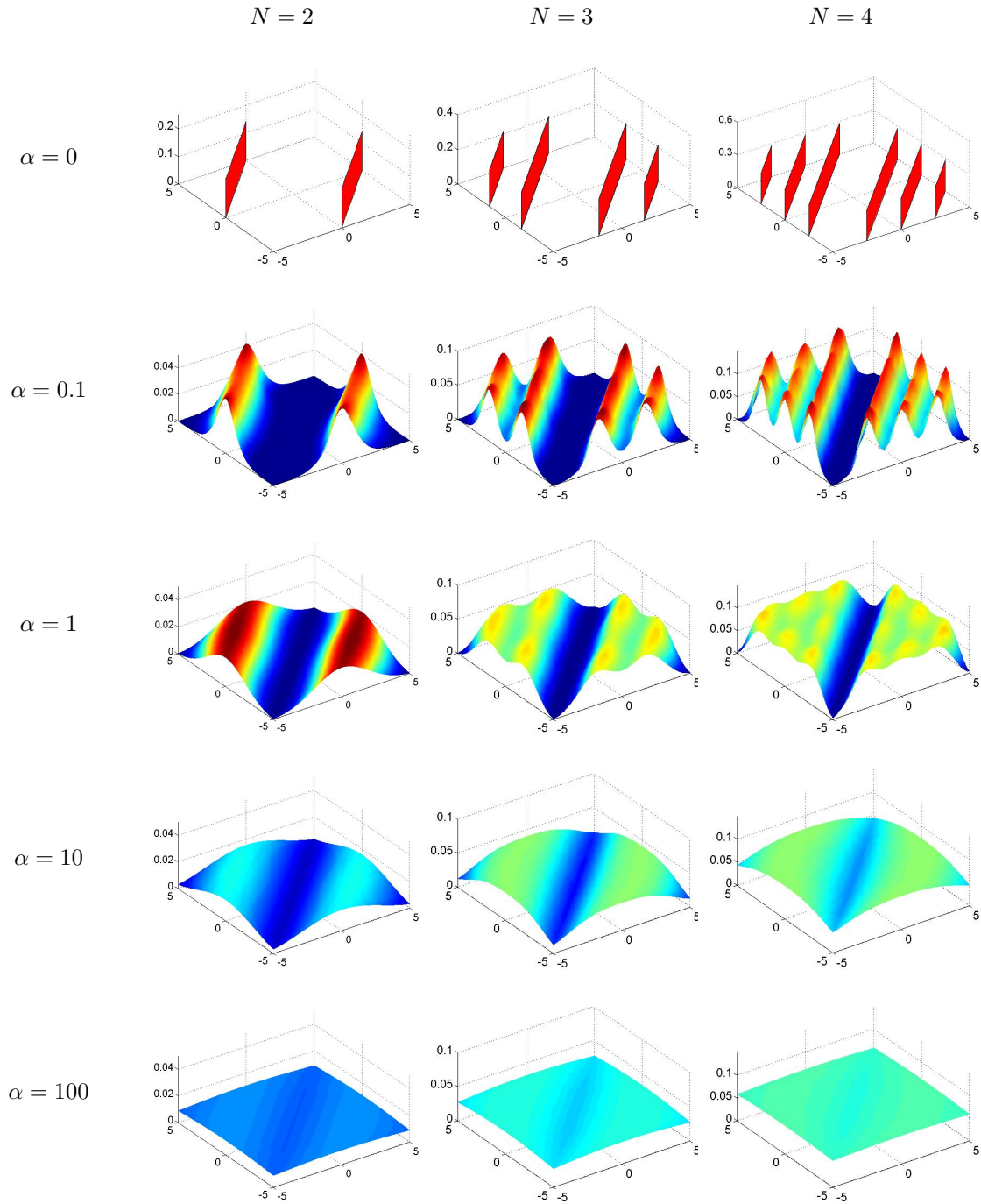

\begin{figure}[htbp]
\centering

\begin{picture}(100,580)
\put(-40,540){\makebox(2,2){$N = 2$}}
\put(80,540){\makebox(2,2){$N = 3$}}
\put(200,540){\makebox(2,2){$N = 4$}}

\put(-140,480){\makebox(2,2){$\alpha = 0$}}
\put(-100,430){\includegraphics[width=4.5cm]{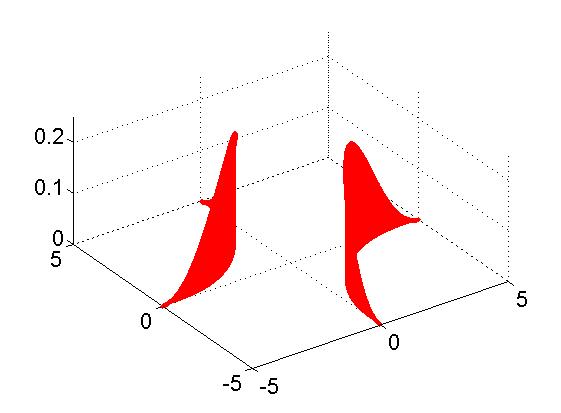}}
\put(20,430){\includegraphics[width=4.5cm]{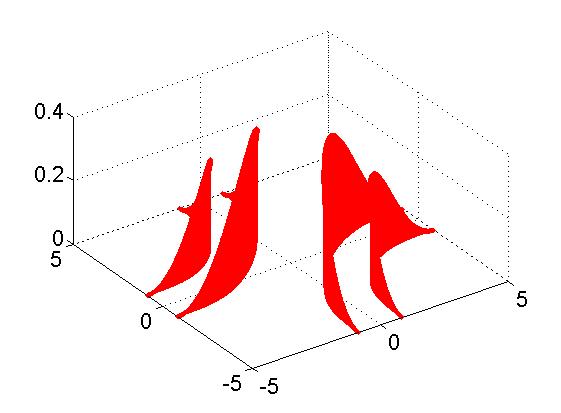}}
\put(140,430){\includegraphics[width=4.5cm]{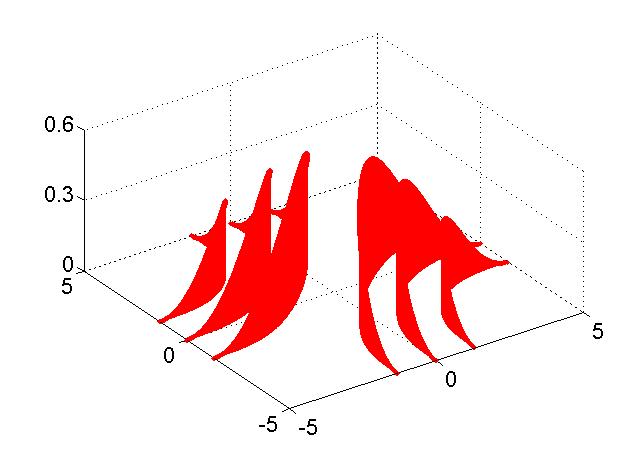}}

\put(-140,380){\makebox(2,2){$\alpha = 0.1$}}
\put(-100,330){\includegraphics[width=4.5cm]{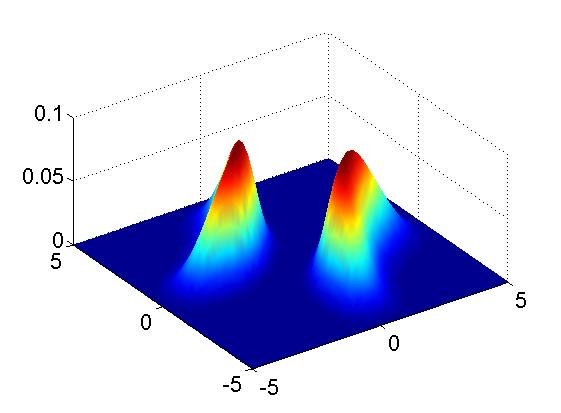}}
\put(20,330){\includegraphics[width=4.5cm]{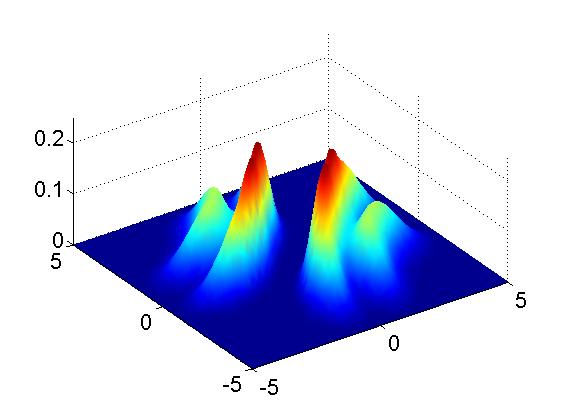}}
\put(140,330){\includegraphics[width=4.5cm]{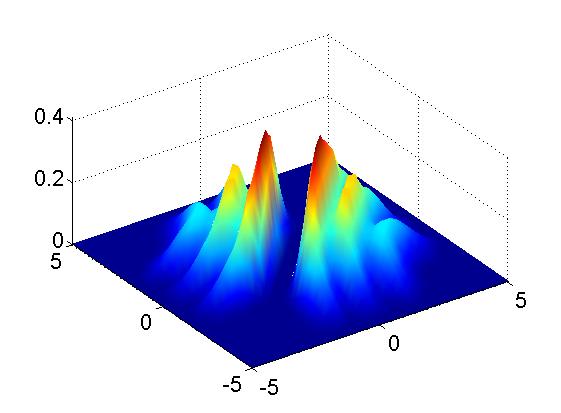}}

\put(-140,280){\makebox(2,2){$\alpha = 1$}}
\put(-100,230){\includegraphics[width=4.5cm]{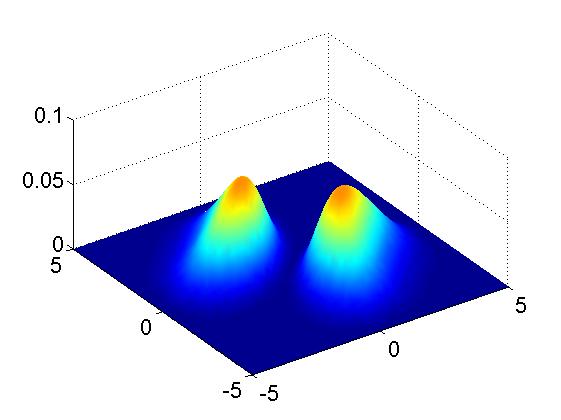}}
\put(20,230){\includegraphics[width=4.5cm]{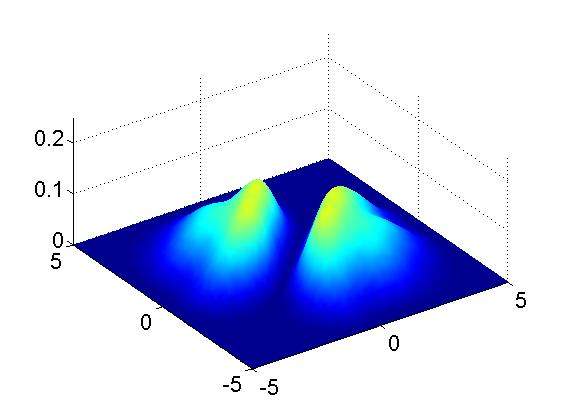}}
\put(140,230){\includegraphics[width=4.5cm]{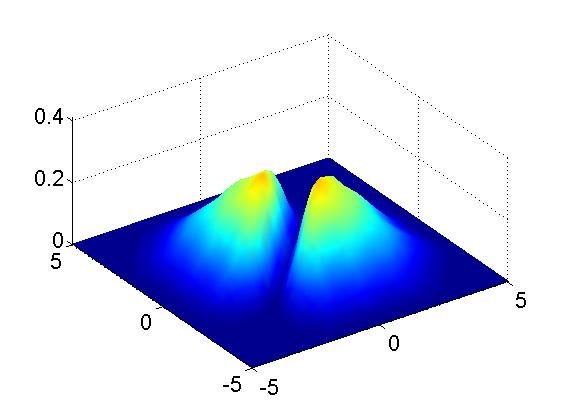}}

\put(-140,180){\makebox(2,2){$\alpha = 10$}}
\put(-100,130){\includegraphics[width=4.5cm]{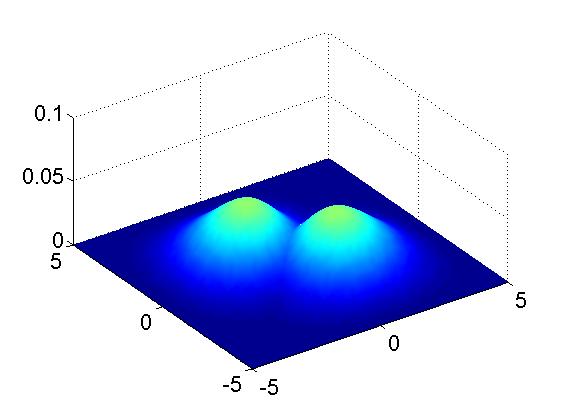}}
\put(20,130){\includegraphics[width=4.5cm]{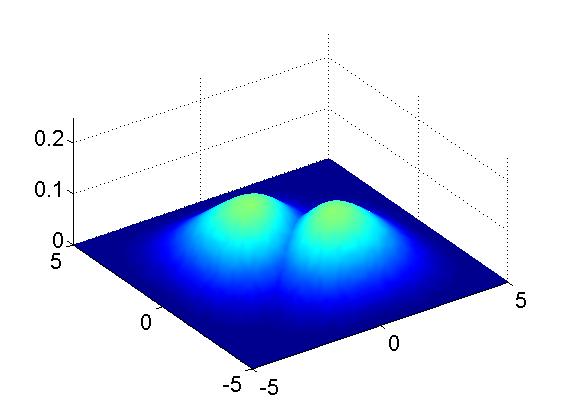}}
\put(140,130){\includegraphics[width=4.5cm]{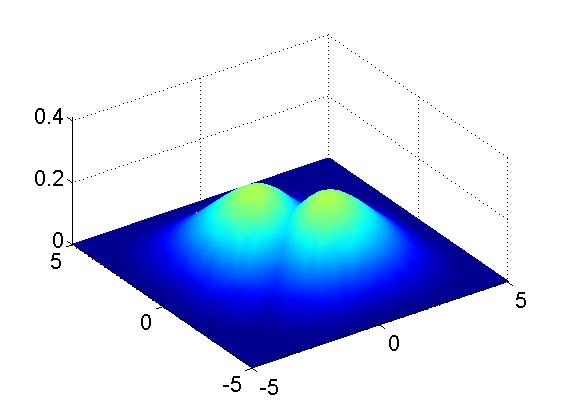}}

\put(-140,80){\makebox(2,2){$\alpha = 100$}}
\put(-100,30){\includegraphics[width=4.5cm]{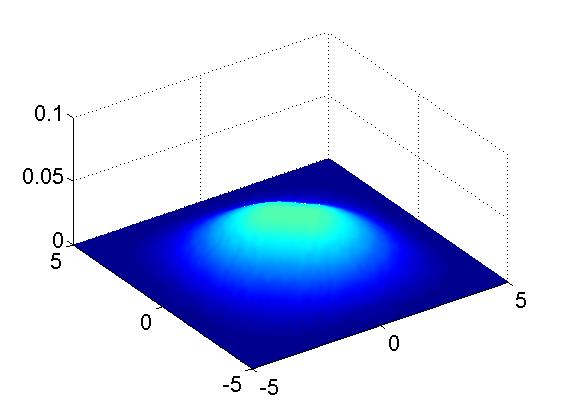}}
\put(20,30){\includegraphics[width=4.5cm]{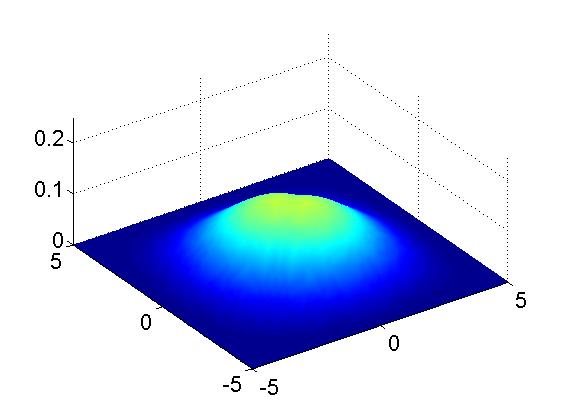}}
\put(140,30){\includegraphics[width=4.5cm]{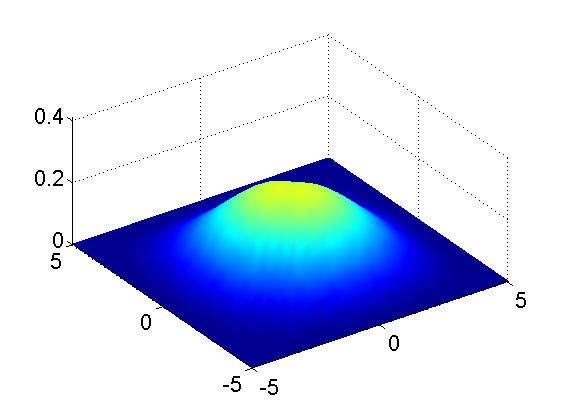}}

\end{picture}
\caption{(bosons) Pair densities with 
$\rho(x)\!=\!\frac{N}{2L}\left(1\! +\!\cos(\frac{\pi}{L}x)\right)\;$ (eqs.~\eqref{SPD},~\eqref{ANDir}, but with spinless symmetric $\Psi$). Top row: SCE/optimal transport (based on exact results \cite{seidl99a}).}
\label{fig-bos-convex}
\end{figure}

We observe that when $\alpha$ is small (e.g., $\alpha=0.1$), the pair densities
are highly localized as $2(N-1)$ ridges around ${\rm supp}(\rho_2^{SCE})$.
As $\alpha$ increases, the pair densities are smoothed out gradually.
The $2(N-1)$ ridges are still visible when $\alpha=1$ but merge with each other when $\alpha=$10 and 100.
The profiles of the pair densities strongly reflect the number of particles
(particularly when $\alpha$ is small), a phenomenon that is missed by the standard DFT models.
When $\alpha$ equals 100, the pair densities are very close to the statistically independent
function $\frac{1}{2}\left(1-\frac{1}{N}\right)\rho(x)\rho(y)$ predicted in Proposition \ref{propo-Falpha-3}.
In fact, the behavior of the pair densities as $\alpha$ goes from 0 to infinity can be viewed as
a process in which the Coulomb holes fade away and the correlations are smoothed out towards statistical independence.

Moreover, we plot the Lagrange multipliers $\lambda(x)$ for systems with 4 bosons
and different values of $\alpha$ in Figure \ref{fig-lambda-boson}.
We have mentioned that $-\lambda(x)$ can be viewed as the external potential that
has $\rho$ as the ground state density.
Therefore, shifting $\lambda(x)$ by an additive constant makes no difference,
and we can use an appropriate shift to allow better comparisons in the picture.
When $\alpha=0$, the SCE Lagrange multiplier can be calculated according to the formulae in \cite{seidl07,ChenEtAl}.
When $\alpha$ is small, the potentials are actually quite close to the SCE case.
As $\alpha$ increases, the potentials converge to constant functions for homogeneous systems
and become steeper and steeper for inhomogeneous systems to cancel the kinetic energy and constrain the particles.

\begin{figure}[ht]
\centering
\includegraphics[height=4.5cm]{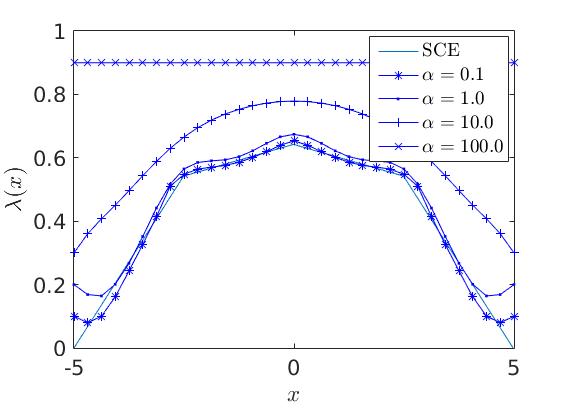}
\includegraphics[height=4.5cm]{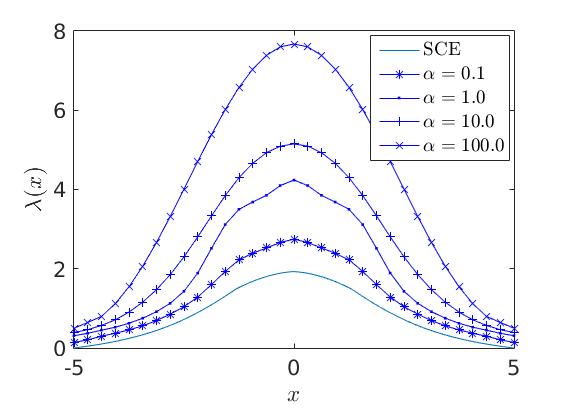}
\caption{The Lagrange multipliers $\lambda(x)$ for systems with 4 bosons.
Left: $\rho(x)\equiv\frac{N}{2L}$. Right: $\rho(x)=\frac{N}{2L}\left(1+\cos(\frac{\pi}{L}x)\right)$.}
\label{fig-lambda-boson}
\end{figure}

\subsection{Fermions}\label{sec-fermion}

Let us now come back to the fermions with spin variables and antisymmetry constraint in $\AN$.
For simplicity, we use the notations $\uparrow$ and $\downarrow$ for spin up and spin down, respectively.
With the same partition $\mathcal{T}$ of $\Omega$ as that in Section \ref{sec-boson},
the single-particle basis set becomes
\begin{eqnarray*}
\{\chi_{j,s}(x,\sigma):~j=1,\cdots,m,~s=\uparrow,\downarrow\},
\end{eqnarray*}
where $x\in\Omega$, $\sigma\in\{\uparrow,\downarrow\}$, and $\chi_{j,s}(x,\sigma)=\chi_j(x)$ if $s=\sigma$ and $0$ otherwise. The classical product \eqref{basis-tensor} needs to be replaced by the Slater determinant
$\psi_{{\bf j},{\bf s}}({\bf x},{\pmb \sigma})$ of the $N$ one-body basis functions $\chi_{j_1,s_1},...,\chi_{j_N,s_N}$, 
with ${\bf j}=(j_1,\cdots,j_N)\in\{1,\cdots,m\}^N$ and ${\bf s}=(s_1,\cdots,s_N)\in\{\uparrow,\downarrow\}^N$.
The number of degrees of freedom is $\binom{2m}{N}$. We denote by $\mathcal{V}_m^N$
the $N$-fermion space spanned by the basis functions $\{\psi_{{\bf j}, {\bf s}}\}$.

The corresponding variational formulation of \eqref{eq-euler-lag} now reads as follows:
Find $\lambda\in V_m$ and $\Psi\in\mathcal{V}_m^N$ such that
\begin{eqnarray}\label{eq-euler-fermion-dis}
\left\{\begin{array}{l} \displaystyle
\frac{\alpha}{2}(\nabla\Psi,\nabla v) + \sum_{1\leq i<j\leq N}(c(|x_i-x_j|)\Psi,v)
=\sum_{i=1}^N(\lambda(x_i)\Psi,v)\quad \forall~v\in\mathcal{V}_m^N \\ [1ex]
\begin{array}{r}
\displaystyle \rho(x) = N\sum_{{\pmb\sigma}\in\{\uparrow,\downarrow\}^N}
\int_{\Omega^{N-1}}|\Psi(x,\sigma_1,x_2,\sigma_2,\cdots,x_N,\sigma_N)|^2 dx_2\cdots dx_N \quad\quad\\
\quad{\rm with}~ x=a^1,a^2,\cdots,a^m \end{array}
\end{array}\right..\quad
\end{eqnarray}
Using Algorithm \ref{algorithm-newton} with $\varepsilon=10^{-4}$ and $m=40$ for $N=2,3$, $m=32$ for $N=4$, we calculate $\lambda$ and $\Psi$ for homogeneous and inhomogeneous electron densities given by \eqref{rho-hom} and \eqref{rho-inhom}.
The ground state pair densities are depicted in Figures \ref{fig-fermion-periodic} and \ref{fig-fermion-convex}.
As an illustration of the efficiency and stability of Algorithm \ref{algorithm-newton},
we present a convergence curve of $\|\rho_0-\rho_k\|$ in Figure \ref{fig-newton}.

\begin{figure}[htbp]
\centering

\begin{picture}(100,580)
\put(-40,540){\makebox(2,2){$N = 2$}}
\put(80,540){\makebox(2,2){$N = 3$}}
\put(200,540){\makebox(2,2){$N = 4$}}

\put(-140,480){\makebox(2,2){$\alpha = 0$}}
\put(-100,430){\includegraphics[width=4.5cm]{figs/2_sce_per.jpg}}
\put(20,430){\includegraphics[width=4.5cm]{figs/3_sce_per.jpg}}
\put(140,430){\includegraphics[width=4.5cm]{figs/4_sce_per.jpg}}

\put(-140,380){\makebox(2,2){$\alpha = 0.1$}}
\put(-100,330){\includegraphics[width=4.5cm]{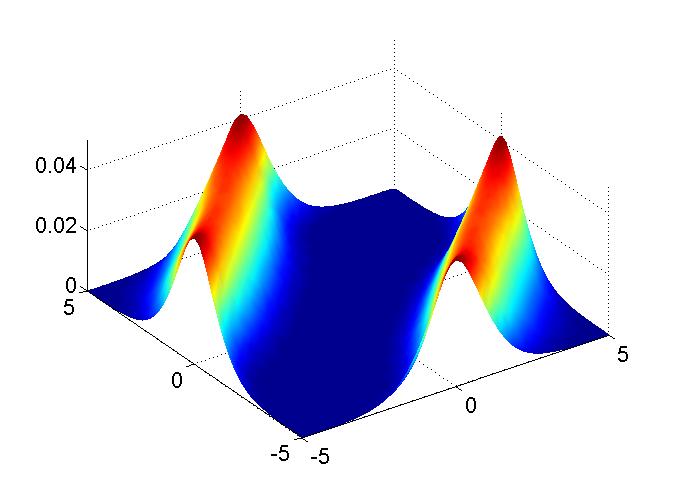}}
\put(20,330){\includegraphics[width=4.5cm]{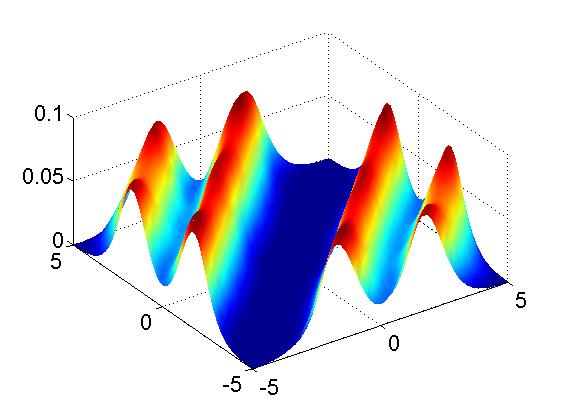}}
\put(140,330){\includegraphics[width=4.5cm]{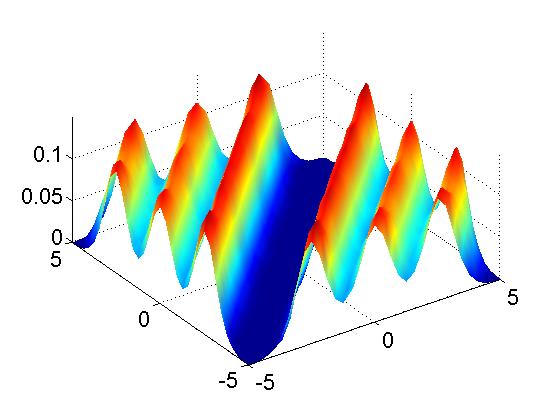}}

\put(-140,280){\makebox(2,2){$\alpha = 1$}}
\put(-100,230){\includegraphics[width=4.5cm]{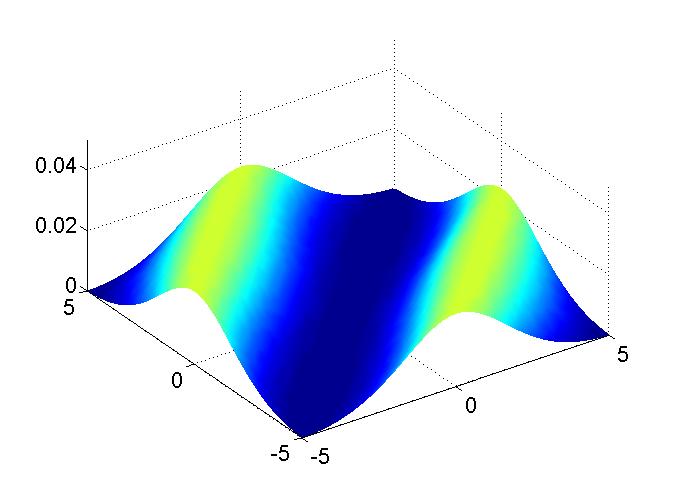}}
\put(20,230){\includegraphics[width=4.5cm]{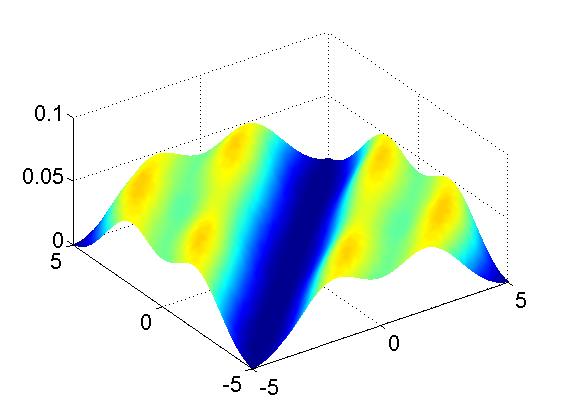}}
\put(140,230){\includegraphics[width=4.5cm]{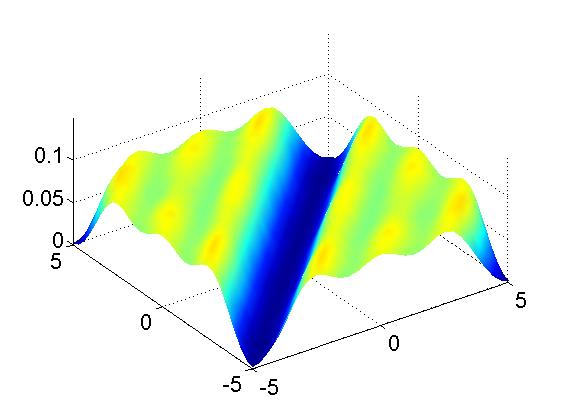}}

\put(-140,180){\makebox(2,2){$\alpha = 10$}}
\put(-100,130){\includegraphics[width=4.5cm]{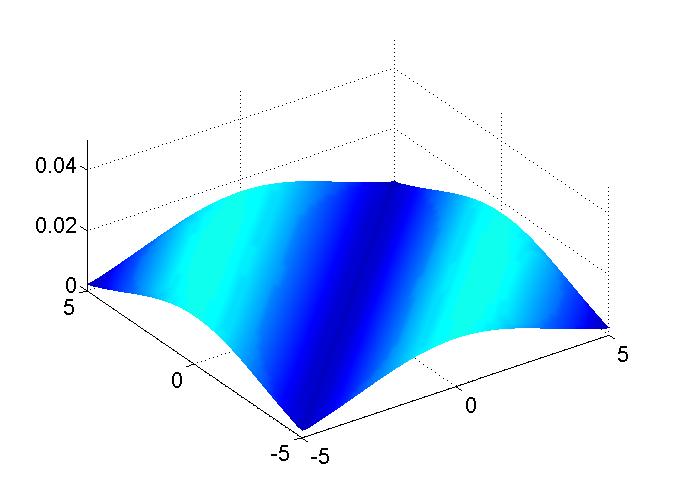}}
\put(20,130){\includegraphics[width=4.5cm]{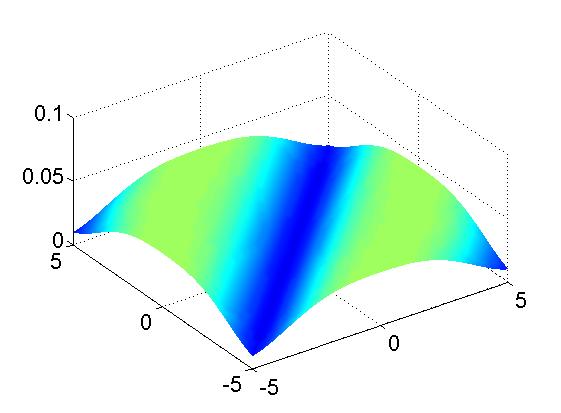}}
\put(140,130){\includegraphics[width=4.5cm]{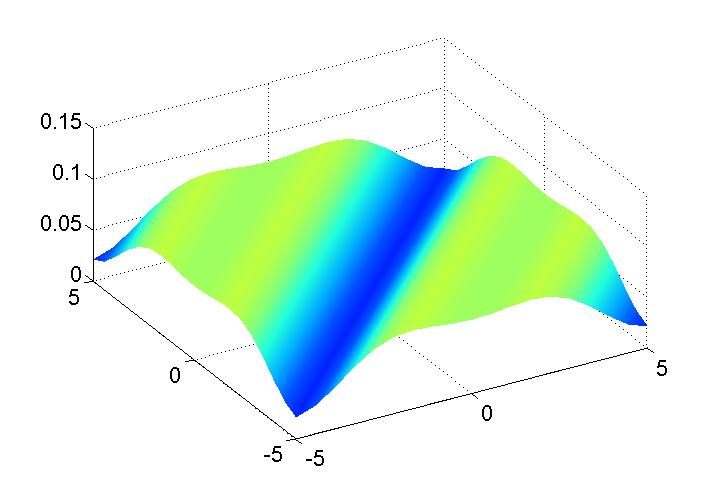}}

\put(-140,80){\makebox(2,2){$\alpha = 100$}}
\put(-100,30){\includegraphics[width=4.5cm]{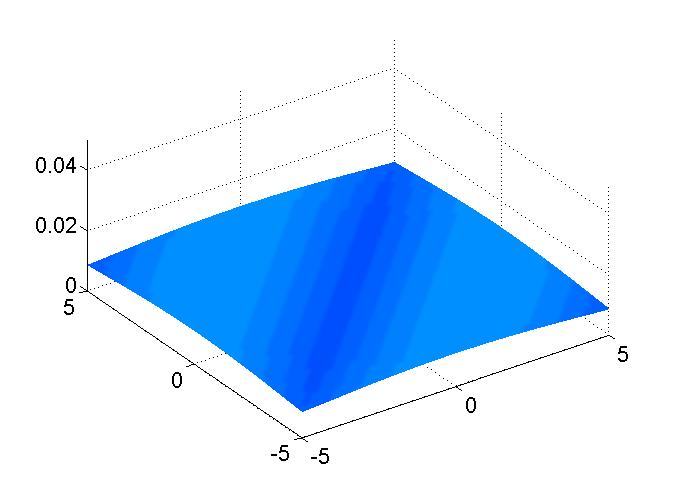}}
\put(20,30){\includegraphics[width=4.5cm]{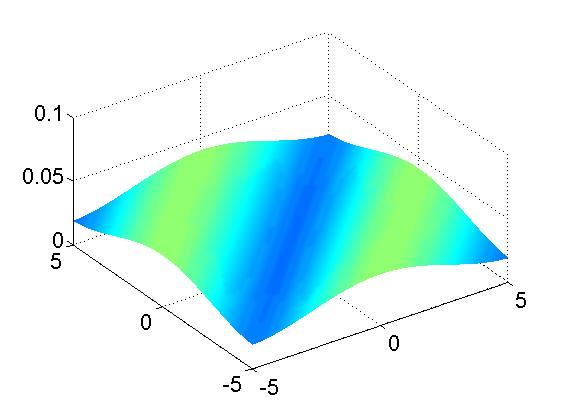}}
\put(140,30){\includegraphics[width=4.5cm]{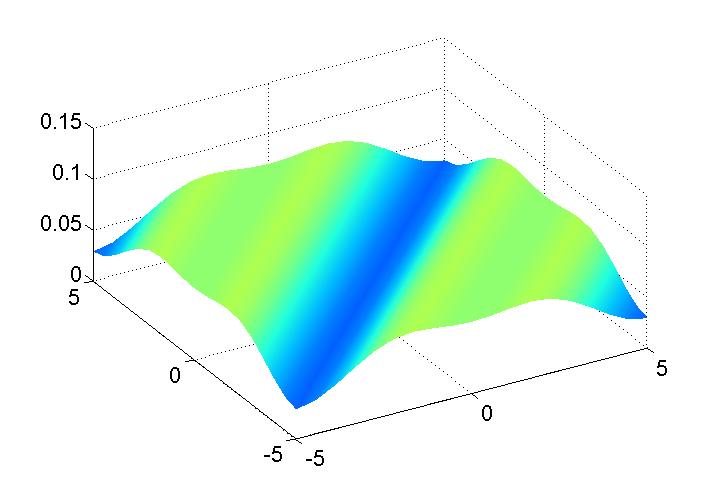}}

\end{picture}
\caption{(fermions) Pair densities with $\rho(x)\equiv\frac{N}{2L}\;$ (eqs.~\eqref{SPD},~\eqref{ANper}). Top row: SCE/optimal transport (based on exact results \cite{seidl99a}, see also \cite{cotar13a,colombo}).}
\label{fig-fermion-periodic}
\end{figure}

\begin{figure}[htbp]
\centering

\begin{picture}(100,580)
\put(-40,540){\makebox(2,2){$N = 2$}}
\put(80,540){\makebox(2,2){$N = 3$}}
\put(200,540){\makebox(2,2){$N = 4$}}

\put(-140,480){\makebox(2,2){$\alpha = 0$}}
\put(-100,430){\includegraphics[width=4.5cm]{figs/2_sce_cos.jpg}}
\put(20,430){\includegraphics[width=4.5cm]{figs/3_sce_cos.jpg}}
\put(140,430){\includegraphics[width=4.5cm]{figs/4_sce_cos.jpg}}

\put(-140,380){\makebox(2,2){$\alpha = 0.1$}}
\put(-100,330){\includegraphics[width=4.5cm]{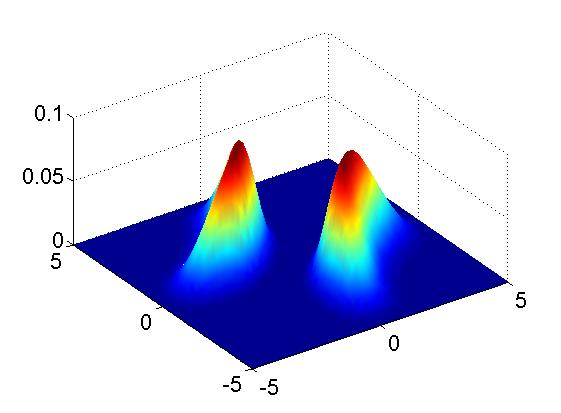}}
\put(20,330){\includegraphics[width=4.5cm]{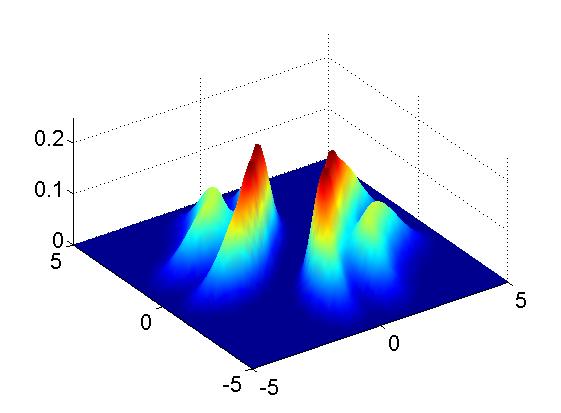}}
\put(140,330){\includegraphics[width=4.5cm]{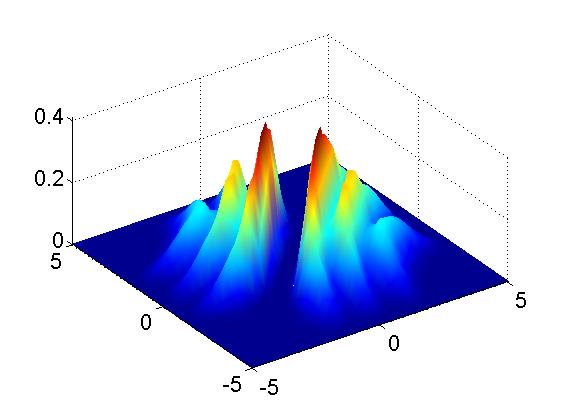}}

\put(-140,280){\makebox(2,2){$\alpha = 1$}}
\put(-100,230){\includegraphics[width=4.5cm]{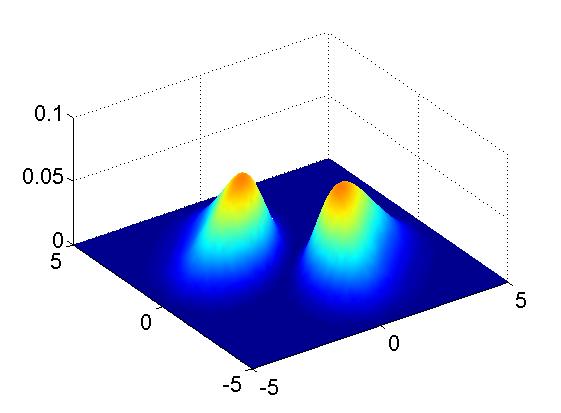}}
\put(20,230){\includegraphics[width=4.5cm]{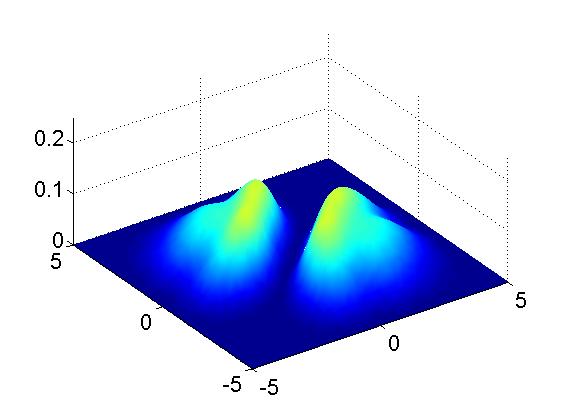}}
\put(140,230){\includegraphics[width=4.5cm]{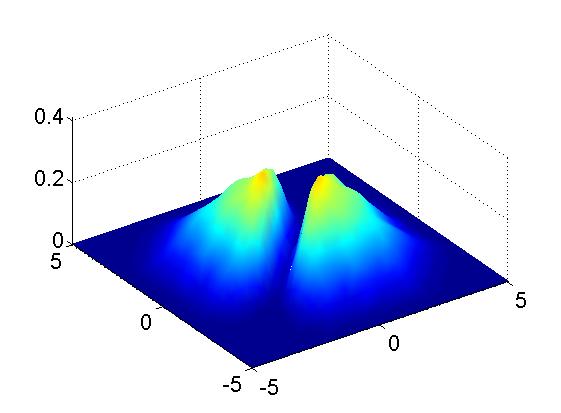}}

\put(-140,180){\makebox(2,2){$\alpha = 10$}}
\put(-100,130){\includegraphics[width=4.5cm]{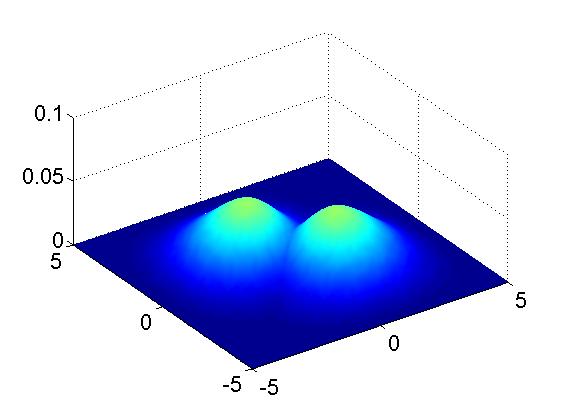}}
\put(20,130){\includegraphics[width=4.5cm]{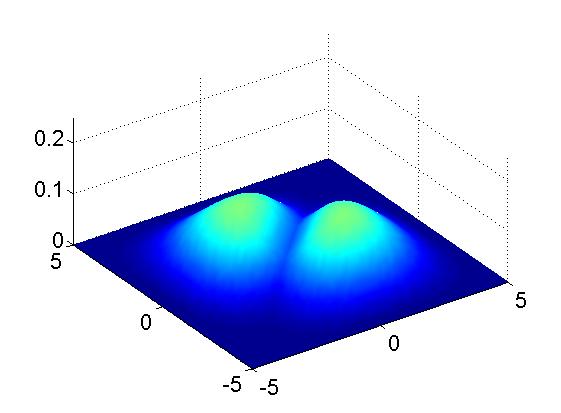}}
\put(140,130){\includegraphics[width=4.5cm]{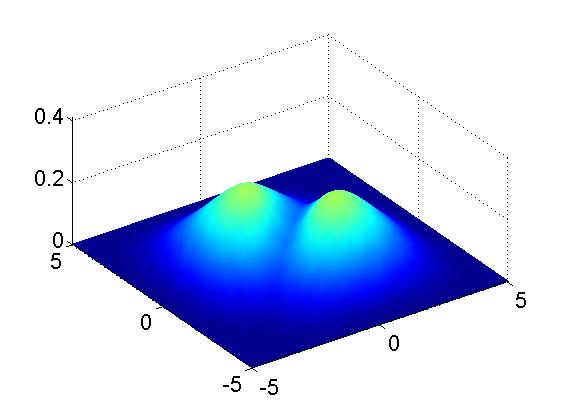}}

\put(-140,80){\makebox(2,2){$\alpha = 100$}}
\put(-100,30){\includegraphics[width=4.5cm]{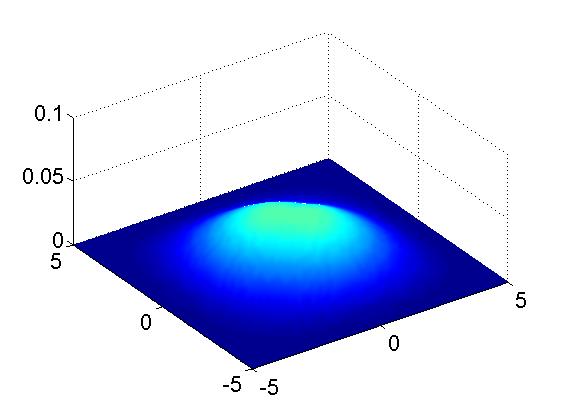}}
\put(20,30){\includegraphics[width=4.5cm]{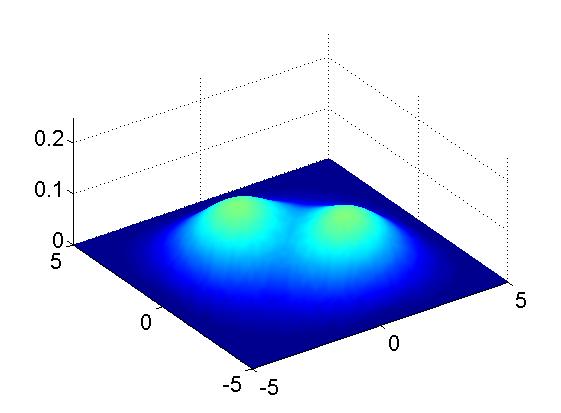}}
\put(140,30){\includegraphics[width=4.5cm]{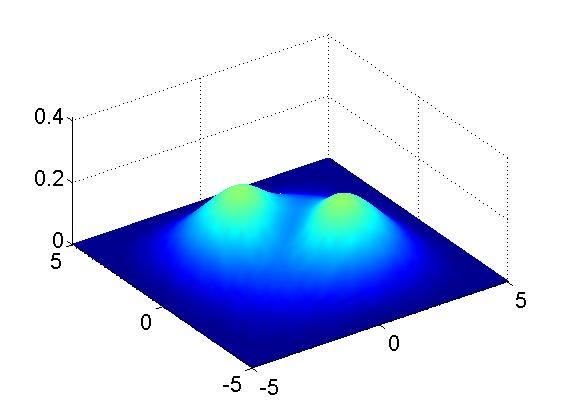}}

\end{picture}
\caption{(fermions) Pair densities with $\rho(x)=\frac{N}{2L}\left(1+\cos(\frac{\pi}{L}x)\right)\;$ (eqs.~\eqref{SPD},~\eqref{ANDir}).  Top row: SCE/optimal transport (based on exact results \cite{seidl99a}, see also \cite{cotar13a,colombo}).}
\label{fig-fermion-convex}
\end{figure}

First of all, when $N=2$, we have the same pair densities as those of bosons.
This is easy to understand since for two-particle systems, the two spin variables are always paired up, and the antisymmetry constraint does not affect the spatial variables.

We also find similar pair densities for bosons and fermions when $\alpha$ is small (e.g., $\alpha=0.1$).
In this case, the particles are strongly correlated to each other for both bosons and fermions,
and are always localized in different regions of space that have very little overlap.
Therefore, the pair densities are almost independent of the choice of the spin variables:
both the symmetric and antisymmetric choice give very similar spatial distributions,
and the particles do not sense very much whether they are fermions or bosons.
From the pictures, we can also draw some similar conclusions as those for bosons:
When $\alpha$ is small, the particle number can be recovered by counting the number of ridges of the pair densities. 
As $\alpha$ increases, the $2(N-1)$ ridges merge together.

A significant difference between bosons and fermions is that,
when $\alpha$ goes towards infinity, the pair densities of fermions do not become statistically independent if $N>2$, but are depleted near the diagonal $x=y$, a phenomenon known as ``exchange holes'' (see e.g. \cite{martin05}). 
As $\alpha$ increases, the effects of Coulomb repulsion get weaker and weaker and the Coulomb holes
are fading out, whilst the exchange holes take over.
For comparison, the theoretical $\rho_2$ as $\alpha\to\infty$ (for homogeneous $\rho$ with $N=4$) 
is plotted in Figure \ref{fig-slater}. It corresponds to a Hund's rule selection from the degenerate ground state of $T$ (see Theorem \ref{ThmB}), consisting of the orbitals $0\!\uparrow$, $0\!\downarrow$, $1\!\uparrow$, $(-1)\!\uparrow$ (in the notation \eqref{orbdef}).
We observe that it is extremely close to the numerically computed pair density at large $\alpha$ (shown for $\alpha=100$ in Fig.~\ref{fig-fermion-periodic}, bottom right panel).

\begin{figure}[ht]
\begin{minipage}[t]{0.5\linewidth}
\centering
\includegraphics[width=6.0cm]{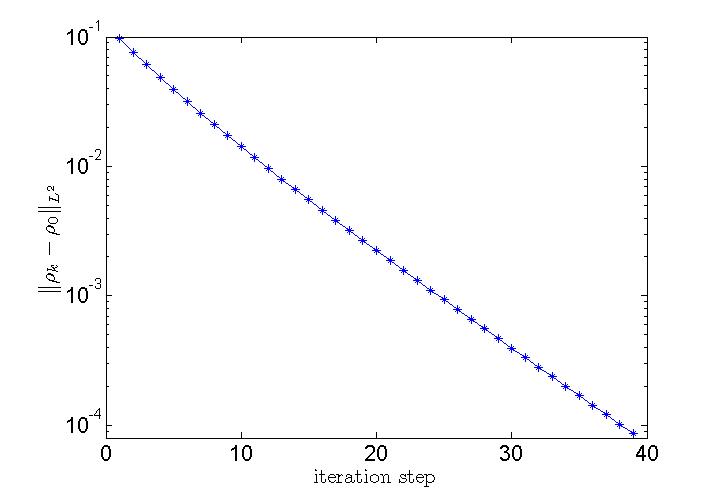}
\caption{Convergence curve of $\|\rho_k-\rho_0\|$ of Algorithm \ref{algorithm-newton} for $N=2$ and $\alpha=0.1$.}
\label{fig-newton}
\end{minipage}
\hskip 0.3cm~
\begin{minipage}[t]{0.5\linewidth}
\centering
\includegraphics[width=6.0cm]{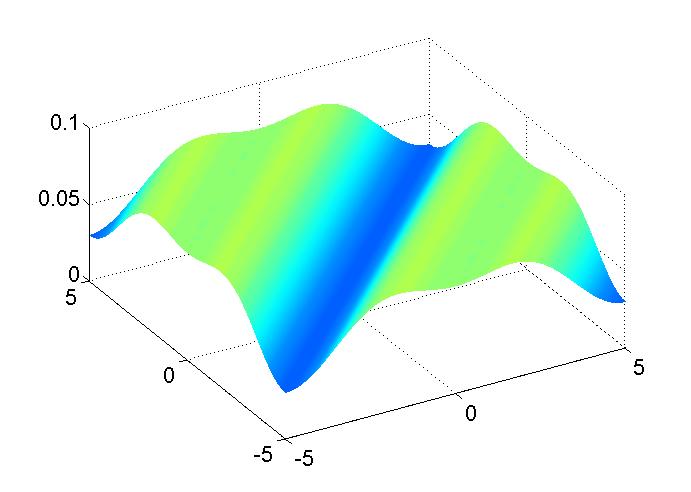}
\caption{Theoretical pair density as $\alpha\to\infty$, coming from the spin-polarized Slater determinant
$|0\!\uparrow 0\!\downarrow 1\!\uparrow (\!-\! 1)\!\uparrow\rangle$ (Theorem \ref{ThmB}).}
\label{fig-slater}
\end{minipage}
\end{figure}

The Lagrange multipliers for systems with 4 fermions are presented in Figure \ref{fig-lambda-fermion}.
In comparison with those of bosons, they also converge to a constant potential as $\alpha$ increases
for homogeneous systems, and have a steeper potential at the same value of $\alpha$ for inhomogeneous systems.

\begin{figure}[ht]
\centering
\includegraphics[height=4.5cm]{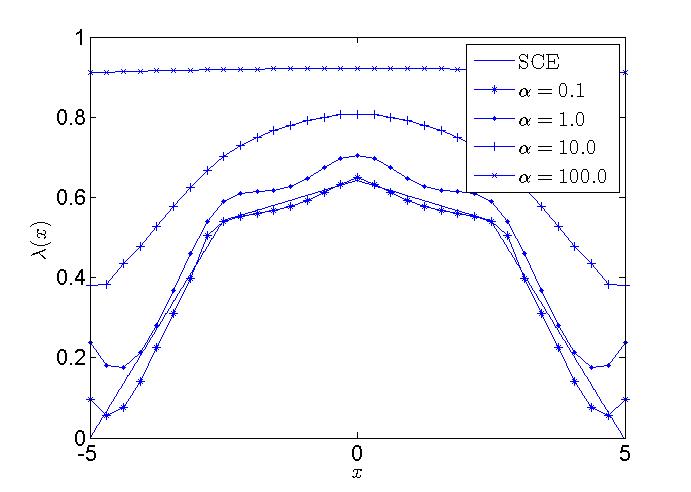}
\includegraphics[height=4.5cm]{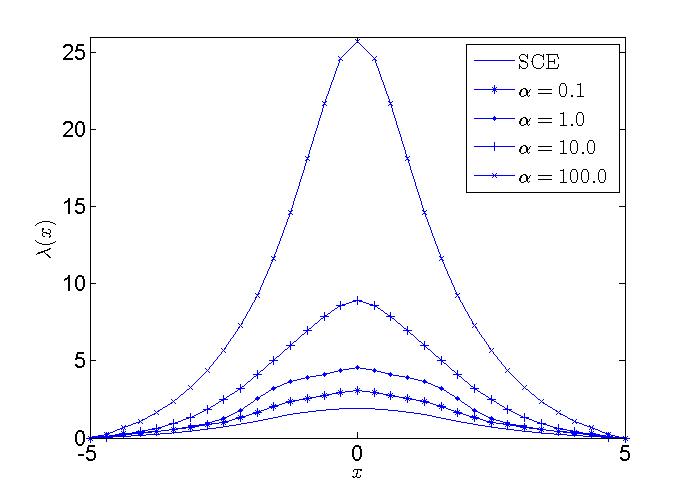}
\caption{The Lagrange multipliers $\lambda(x)$ for systems with 4 fermions.
Left: $\rho(x)\equiv\frac{N}{2L}$. Right: $\rho(x)=\frac{N}{2L}\left(1+\cos(\frac{\pi}{L}x)\right)$.}
\label{fig-lambda-fermion}

\end{figure}

From the numerical simulations in Section \ref{sec-boson} and \ref{sec-fermion}, we conclude
that the pair densities across the whole range of coupling constants are deformed versions of the two limit cases $\alpha=0$ and $\alpha=\infty$, with a slow and steady cross-over and without any additional effects appearing. The ``information'' in the pair densities for all $\alpha$ can somehow be recovered from just the two end values $\alpha=0$ and $\alpha=\infty$. By contrast, none of the end-value pair densities gives useful information about what happens at the other end. This lends theoretical support to the idea in \cite{Seidl2000} of two-end interpolation functionals. It should be very interesting to try to relate the specific functional proposed there to an underlying pair density model and compare to a theoretical adiabatic connection curve.

Let us also emphasize the strong pair density localization without single-particle localization and the strong $N$-dependence. The latter is missed completely by the local density approximation (LDA), which is based on uniform electron gas theory ($N=\infty$). As regards the former effect, it is not clear (at least to the authors) to what extent it is accounted for by {\it any} of the models used in practice. For homogeneous $\rho$ and large $N$, the true pair density profile is captured implicitly through use of the LDA correlation energy; but we do not know what happens implicitly to the pair density when applying, say, the LDA or gradient corrections or a fraction of exact exchange to a typical inhomogeneous $\rho$. See Section \ref{sec-ansatz} for further discussion.

To end this section, we summarize some of the characteristics of the pair densities in Table \ref{table-sum}.

\begin{table}[ht]
\centering
\begin{tabular}{|c|c|c|}\hline
       pair densities   & {\bf bosons}    &  {\bf fermions}        \\\hline
       $\alpha=0$       & SCE         & SCE               \\\hline
       $\alpha=\infty$  & statistical independence    &  single slater determinant      \\\hline
       \multirow{2}{*}{Coulomb holes}    & yes         & yes                \\
       & fade out as $\alpha$ increases    & fade out as $\alpha$ increases \\\hline
       \multirow{2}{*}{exchange holes}   & no          & yes                \\
       &    &   fade out as $\alpha$ decreases   \\\hline
       $N$-dependence       & yes         & yes               \\\hline
\end{tabular}
\caption{Summary of the characteristics of pair densities.}
\label{table-sum}
\end{table}

\section{Rigorous asymptotic results} \label{sec:asy}
The following asymptotic results in 1D support our numerical findings, and were used to test the correctness of our code. Results of this type are well-known in the physics literature (except perhaps those on ``selection rules'' which emerge in the non-interacting limit in case of orbital degeneracies) and the novelty consists only in providing rigorous proofs. The reader is reminded that on the rigorous level very little is known about exact DFT and even basic issues as raised in \cite{lieb83} such as continuity of the HK functional remain open. 

Recall from Section \ref{sec-hk} the scaled density-to-pair-density map $\rho_{2,\alpha}[\rho]=D_{\alpha^{-1}}\rho_2[D_\alpha\rho]$, where $\rho_2$ is the original density-to-pair-density map. 

\begin{theorem} \label{ThmA} (Small $\alpha$ limit, 1D systems) Let $\rho$ be any single-particle density on $\R$ belonging to the class $\RN$ (see \eqref{RN}), $N\ge 2$.
Assume that $\rho>0$ in some finite or infinite interval $(a,b)$, and $\rho=0$ outside. Let 
$T_1,..,T_N$ be the following optimal transport maps found in \cite{seidl99a} and justified rigorously in \cite{cotar13a} for $N=2$ and in \cite{colombo} for general $N$: let $d_0=a<d_1<...<d_{N-1}<d_N=b$ be the partition of $(a,b)$ into $N$ sub-intervals of equal mass, i.e. 
$$    \int_{d_{i-1}}^{d_{i}} \rho = 1 \;\;\;(i=1,...,N),      $$
and let $T_2$ be the unique $\rho$-preserving map which monotonically maps each interval
$[d_{i-1},d_i]$ (i=1,...,N-1) to the next interval $[d_i,d_{i+1}]$ and the last interval $[d_{N-1},d_N]$ to the first, $[d_0,d_1]$. Let $T_1(x)=x$, and let $T_j$, $j=3,..,N$, be the (j-1)-fold composition of $T_2$ with itself. (See Figures \ref{fig-rho_per}, \ref{fig-rho_cos}.) Then 
$$
    \lim_{\alpha\to 0} \rho_{2,\alpha}[\rho] = \frac12 \sum_{j=2}^N \frac{\rho(x)}
    {\sqrt{1 + T_j'(x)^2}} ds \Big|_{y=T_j(x)},
$$
the limit being in the sense of weak* convergence of Radon measures. 
\end{theorem}

\begin{proof} Let $\Psi_\alpha$ be a minimizer of the variational problem in \eqref{SPD}, and let $\rho_{N,\alpha}(x_1,..,x_N)=\sum_{s_1,..,s_N\in\Z_2} |\Psi_\alpha(x_1,s_1,..,x_N,s_N)|^2$. Since the $\rho_{N,\alpha}$ have marginal $\rho$, they are a tight family of probability measures (to show this one proceeds analogously to the proof of a similar result in the appendix of \cite{lieb83}) and hence possess a subsequence (see \cite{cotar13a}), again denoted $\rho_{N,\alpha}$, converging weak* to a probability measure $\rho_{N,*}$ as $\alpha\to 0$. 
By standard arguments $\rho_{N,*}$ has one-body marginal $\rho$. Moreover, by dropping the kinetic energy from \eqref{SPD} and using the lower semicontinuity of the interaction energy under weak* convergence, and letting $\hat{V}_{ee} = \sum_{i<j}c(x_i-x_j)$,
$$
  \lim_{\alpha\to 0} \inf_{\Psi\mapsto\rho}(\alpha T[\Psi] + V_{ee}[\Psi]) = \lim_{\alpha\to 0} \Bigl(\alpha T[\Psi_\alpha] + \int_{\R^N} 
   \hat{V}_{ee}\rho_{N,\alpha}\Bigr) \ge \int_{\R^N} \hat{V}_{ee} d\rho_{N,*}. 
$$
On the other hand, as proved in \cite{cotar13a} the left hand side equals $\min_{\rho_N} \int \hat{V}_{ee} d\rho_N$, the minimum being over symmetric probability measures on $\R^N$ with marginal $\rho$. It follows that $\rho_{N,*}$ is a minimizer of the latter problem. 
By the results of \cite{seidl99a} as made rigorous in \cite{colombo}, the minimizer of the latter problem is unique and given by \eqref{Monge}, with the above explicit maps $T_1,..,T_N$. The uniqueness implies that the whole sequence $\rho_{N,\alpha}$ converges weak* to $\rho_{N,*}$. Next, this latter convergence implies weak* convergence of the associated two-body density $\rho_{2}^{\Psi_\alpha}$ to the pair density ${N\choose 2}\int \rho_{N,*}dx_2...dx_N$ of $\rho_{N,*}$. The assertion now follows from our result \eqref{rho2geo}. 
\end{proof}
Here and below, we denote the eigenfunctions of the Laplacian on $[-L,L]$ with periodic boundary conditions by 
\begin{equation} \label{orbdef}
	|k\rangle(x) := \frac{1}{\sqrt{2L}} e^{ik\frac{\pi}{L} x} \;\;\; (k\in\Z)
\end{equation}
and the associated spin-orbitals $|k\rangle(x)\delta_{\uparrow}(s)$ and $k\rangle(x)\delta_{\downarrow}(s)$ by $|k\uparrow\rangle$, $|k\downarrow\rangle$. 
\begin{theorem} \label{ThmB} (Large $\alpha$ limit, homogeneous 1D systems) Let $\rho(x)=\bar{\rho}=N/(2L)$ be the homogeneous density on $[-L,L]$, and
let $\rho_{2,\alpha}[\rho]=D_{\alpha^{-1}}\rho_2[D_\alpha\rho]$ be the scaled density-to-pair-density map for periodic boundary conditions on $[-L,L]$ (\eqref{SPD} with $\AN$ given by \eqref{ANper}). Let $\Psi$ be the Slater determinant built from the first $N$ orbitals $\phi_1,\dots,\phi_N$ of the (partially spin-polarized) sequence $|0\uparrow\rangle$, $|0\downarrow\rangle$, $|1\uparrow\rangle$, $|(-1)\uparrow\rangle$, $|1\downarrow\rangle$, 
$|(-1)\downarrow\rangle$, $|2\uparrow\rangle$, $|(-2)\uparrow\rangle$, $|2\downarrow\rangle$, $|(-2)\downarrow\rangle$, $\dots$. 
Then, letting $z=\frac{\pi}{L}(x-y)$, 
\begin{eqnarray*}
   \lim_{\alpha\to\infty} \rho_{2,\alpha}[\rho](x,y) = \rho_2^{\Psi}(x,y) =  
   \begin{cases} 
   \displaystyle      \frac12 \rhobar^2 
        - \frac{1}{(2L)^2} \frac{\sin^2(\frac{N}{4}z)}{\sin^2(\frac12 z)}, 
                               & N\equiv 2 \mod 4 \\[1ex]
    \displaystyle      \frac12 \rhobar^2 
        -\frac12 \frac{1}{(2L)^2} 
\frac{\sin^2(\frac{N-1}{4}z) + \sin^2(\frac{N+1}{4} z)}{\sin^2(\frac12 z)},  
                               & N\equiv 1 \mbox{ or } 3 \mod 4 \\[1ex]
    \displaystyle      \frac12 \rhobar^2 
        -\frac12 \frac{1}{(2L)^2} 
\frac{\sin^2(\frac{N-2}{4}z) + \sin^2(\frac{N+2}{4} z)}{\sin^2(\frac12 z)},   
                               & N\equiv 0 \mod 4, 
              \end{cases}
\end{eqnarray*}
the limit being in the sense of strong convergence in $L^1([-L,L]^2)$. 
\end{theorem}
\begin{proof} 
We first ignore the constraint $\Psi\mapsto\rho$. Let $X_0$ be the ground state of $\hat{T}=-\frac12\Delta$ on $\AN$, let $P_0$ be the orthogonal projector from $L^2$ 
onto $X_0$, let $X'_0$ be the lowest eigenspace of $P_0 \hat{V}_{ee}P_0$ within $X'_0$ (note that $X'_0=X_0$ if $X_0$ is one-dimensional), and let 
$S'_0=\{\Psi\in X'_0 \, : \, \Psi\mapsto\rho\}$. By degenerate first-order perturbation theory,
together with the fact that by the explicit description below $S'_0$ is nonempty, 
\begin{equation}
   \lim_{\alpha\to\infty} \{ \Psi\in\AN \, | \, \Psi \mbox{ minimizes }T+\mbox{ $\frac{1}{\alpha}$} \mbox{ s/to } \Psi\mapsto\rho\} \subseteq S'_0,
\end{equation}
the limit being in the sense of strong $L^2$ convergence. It follows that the set of pair densities $\rho_{2,\alpha}[\rho]$ satisfies $\lim_{\alpha\to\infty}\rho_{2,\alpha}[\rho]\subseteq\{\rho_2^{\Psi}\, : \, \Psi\in S'_0\}$, 
the limit being in the sense of strong $L^1$ convergence (note that the map $\Psi\mapsto\rho_2$ 
is continuous from $L^2(([-L,L]\times\Z_2)^N)$ to $L^1([-L,L]^2)$).  
To complete the proof of the theorem, we need to understand $S'_0$ explicitly. The ground state $X_0$ of $-\frac12\Delta$ on $\AN$ is given by
\begin{equation}
   \mbox{Span} \; |0^2,1^2,(-1)^2,..,K^2,(-K)^2\rangle 
                                \;\;\; \mbox{ if }N = 2\,\;\mbox{mod}\, 4,\, K=\mbox{$\frac{N-2}{4}$}, \label{N2}
\end{equation}
and by 
\begin{eqnarray}
   & & \mbox{Span} \{ |0^2,1^2,(-1)^2,..,(K\! -\! 1)^2,(-(K\! -\! 1))^2,a_1,..,a_d\rangle \, :   \, a_1,..,a_d = \mbox{ any d} \nonumber \\
   & & \hspace*{1cm} 
     \mbox{orbitals from }K\!\uparrow,K\!\downarrow,-\! K\!\uparrow, -\! K\!\downarrow   \}  
                                   \;\;\;\mbox{otherwise}, \label{Nelse}
\end{eqnarray}
where the notation $k^2$ means that the orbitals $|k\!\uparrow\rangle$ and 
$|k\!\downarrow\rangle$ are both present in the Slater determinant and $d$ and $K$ are as follows: $d=3$ and $K=(N-1)/4$ if $N\equiv 1\mod 4$; $d=2$ and $K=N/4$ if $N\equiv 0 \mod 4$; and $d=1$ and $K=(N+1)/4$ if $N\equiv 3\mod 4$. For $N\neq 0 \mod 4$, $X'_0=X_0$. 
But for $N=0\mod 4$, aligning the two spins is favourable because it generates an additional exchange term. This is a manifestation of the empirical Hund's rule. Thus $X'_0$ is given by the subspace of $X_0$ with total spin $S^2=s(s+1)|_{s=1}$, 
\begin{align}
   X'_0 = \mbox{Span}  \{  & |0^2,1^2,(-1)^2,..,(K\! -\! 1)^2, -(K\! -\! 1)^2, K\uparrow,-K\uparrow \rangle, \nonumber \\
                           & |0^2,1^2,(-1)^2,..,(K\! -\! 1)^2, -(K\! -\! 1)^2, K\downarrow,-K\downarrow \rangle, \nonumber \\
                           & \mbox{$\frac{1}{\sqrt{2}}$}(|0^2,..,-(K\! -\! 1)^2, K\uparrow, (-\! K)\downarrow\rangle + |0^2,..,-(K\! -\! 1)^2, K\downarrow, (-\! K)\uparrow\rangle ) \}  \nonumber \\
                           & \mbox{ if }N=0\;{mod}\, 4. \label{N0}
\end{align}
The three states above are the canonical basis states with $S_3=1$, $-1$, and $0$. 

We now take into account the constraint $\Psi\mapsto\rho$, and determine $S'_0$. 
For even $N$, $S'_0$ is the sphere of unit vectors in $X'_0$. For odd $N$, we claim that
\begin{equation}\label{S0'}
   S'_0 = \{\alpha\Psi_1 + \beta\Psi_2 + \gamma\Psi_3 + \delta\Psi_4 \, : \, 
          |\alpha|^2+|\beta|^2+|\gamma|^2+|\delta|^2=1, \; 
          \begin{pmatrix} \alpha \\ \beta \end{pmatrix} \cdot
          \begin{pmatrix} \overline{\gamma} \\ \overline{\delta}\end{pmatrix} = 0\},
\end{equation}
where for $N\equiv 3\mod 4$ the $\Psi_1,..,\Psi_4$ correspond to the four choice of $a_1$ in \eqref{Nelse} in the listed order, and for $N\equiv 1\mod 4$ they correspond to the four choices $K\downarrow (-K)\downarrow K\uparrow$, $K\uparrow (-K)\uparrow K\downarrow$, $K\downarrow (-K)\downarrow (-K)\uparrow$, and $K\uparrow (-K)\uparrow (-K)\downarrow$ of $a_1,a_2,a_3$. 
For, say, the latter $N$'s, the constraint in \eqref{S0'} follows from the fact that
\begin{eqnarray*}
   \rho^{\alpha\Psi_1+...+\delta\Psi_4}(x) & = & const + |\alpha e^{iK\frac{\pi}{L}x} + \gamma e^{-iK\frac{\pi}{L}x}|^2 + |\beta e^{iK\frac{\pi}{L}x} + \delta e^{-iK\frac{\pi}{L}x}|^2 \\
    & = & const + 2 \, \mbox{Re}(\alpha\overline{\gamma} + \beta\overline{\delta}) \cos(2K\frac{\pi}{L}x) - 2 \, \mbox{Im}(\alpha\overline{\gamma} + \beta\overline{\delta}) \sin(2K\frac{\pi}{L}x)
\end{eqnarray*}
and the linear independence of the three functions $\cos$, $\sin$, and $1$. Finally, for each of the four cases of $N$'s, a tedious calculation gives the corresponding pair densities, as well as the fact that these are independent of the coefficients of the wavefunctions in $S'_0$.
\end{proof}

We find the uniqueness of the limiting $\rho_2$'s despite degeneracy of the limiting ground state wavefunctions remarkable. 

\section{An ansatz for homogeneous systems}\label{sec-ansatz}
\setcounter{equation}{0}\setcounter{figure}{0}

Based on the above numerical and asymptotic results, we shall now design a simple ansatz for
the pair density of homogeneous systems which is accurate across the whole range of coupling constants $\alpha$.

If we look at the pair density graphs for homogeneous systems from a specific angle (see
Figure \ref{fig-rotate} for example), we can observe that they are almost uniform functions of $x-y$. 
\begin{figure}[!htb]
\centering
\includegraphics[width=5.0cm]{figs/f4_1_per.jpg} \hskip 2.5cm
\includegraphics[width=5.0cm]{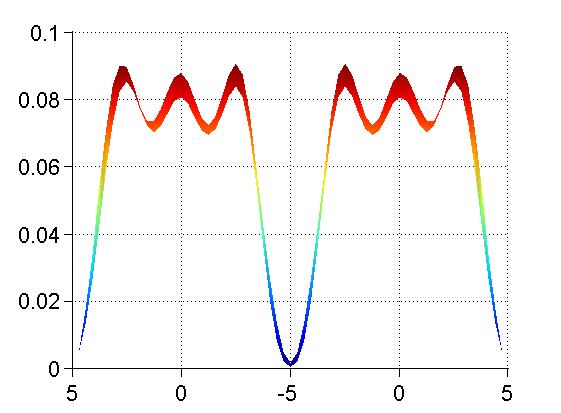}
\put(-185,45){\makebox(10,2){\small view from (-45,0)}}
\put(-185,35){\makebox(5,2){$\longrightarrow$}}
\caption{Rotating the pair density of a homogeneous system with 4 fermions with $\alpha=1$.
Left: view from angle (-35,50). Right: view from angle (-45,0).}
\label{fig-rotate}
\end{figure}
This together with the peaks on the graphs of the transport maps $T_i$ suggests an ansatz of the form
\begin{equation}\label{ansatz-shape}
\rho_2(x,y) \approx c_{\rm n} \left( \sum_{i=2}^{N} \Gamma(d_i(x,y)) \right), \mbox{ where }
d_i(x,y) = \min_{x'}|(x,y)-(x',T_i(x'))|.
\end{equation}
Here $c_{\rm n}$ is a normalization constant and
$\Gamma$ is some shape function. Note that, due to the explicit form of the $T_i$, the above $\rho_2$ depends only on $x-y$. A general formal asymptotic expansion at small $\alpha$ in the physics literature \cite{GVS} or alternatively, in our special case, an elementary calculation detailed below suggests to take $\Gamma$ to be a Gaussian. Thus we make the ansatz
\begin{eqnarray}\label{ansatz-gaussian}
\rho_2(x,y) \approx G^{\rm bos}_{\varsigma}(x,y) =c_{\rm n} \left( \sum_{i=2}^{N} \exp(-\frac{d_i(x,y)^2}{\varsigma^2}) \right)
\end{eqnarray}
where the parameter $\varsigma$ is allowed to depend on the coupling constant $\alpha$ and the particle number $N$. To obtain $\varsigma$, we minimize the $L^1$-error $\|\rho_2-G^{\rm bos}_{\varsigma}\|_{L^1(\Omega^2)}$, where $\rho_2$ is the correct pair density as computed in Section \ref{sec-numerical}. 

See Table \ref{table-appro-boson} for the optimal parameters $\varsigma$ as well as the error (in different norms, calculated by using the finite element discretizations used in Section \ref{sec-numerical}) between the correct pair densities and the ansatz \eqref{ansatz-gaussian}.
We present some cross sections (on $x=-y$) of the pair densities and our ansatz
in Figure \ref{fig-cross-appro-boson}. It appears that the ansatz \eqref{ansatz-gaussian} provides quite an accurate approximation.
Note that the ansatz \eqref{ansatz-gaussian} is accurate at the two limits (by taking $\varsigma=0$ at $\alpha=0$) and ($\varsigma=\infty$ at $\alpha=\infty$), and we can observe from Table \ref{table-appro-boson} that the approximations are better in the regimes where $\alpha$ is very small or large. 

Finally, we give the promised elementary argument which lends theoretical support to our Gaussian ansatz. For $\alpha=0$, $c(r)=1/r$, and, say, $N=2$, the Lagrange multiplier in eq. (6.7) is known exactly and equals $\lambda(x)=|x|/L^2$. Hence the total potential in (6.7) is 
$$
    V(x,y) = \frac{1}{|x-y|} + \frac{|x|}{L^2} + \frac{|y|}{L^2}.
$$ 
This potential is minimal on graph $T_2 = \{ x-y=\pm L\}$. For nonzero but small $\alpha$, the ground state should still be localized near graph $T_2$, and hence we may replace $V(x,y)$ by its second order Taylor polynomial at the nearest point to $(x,y)$ on graph $T_2$. This Taylor approximation is easily calculated to be
$$
    \tilde{V}(x,y) = \frac{2}{L} + \frac{d_2(x,y)^2}{L^3} = \frac{2}{L} + \frac{\min\{(x-y-L)^2,\, (x-y+L)^2\} }{L^3}.
$$ 
Eq. (6.14) with this potential is solved {\it exactly} by a Gaussian of form $e^{-d_2(x,y)^2/const}$, except on the diagonal $x=y$, where the Gaussian and the exact solution should both be small and hence close to each other. This suggests that eq. \eqref{ansatz-gaussian} (with $N=2$) is a good global approximation to the pair density. Giving a rigorous version of this argument is an interesting open problem. 

\begin{small}
\begin{table}[ht]
\centering
\begin{tabular}{|c|c|c|c|c|c|c|}\hline
       $N$  &  $\alpha$  &  optimal $\varsigma^2$  &  $\|\rho_2-G^{\rm bos}_{\varsigma}\|_{L^1}$
       &  $\|\rho_2-G^{\rm bos}_{\varsigma}\|_{L^2}$  &    $\tilde{V}_{ee}[\rho_2]$
       &  $\tilde{V}_{ee}[\rho_2]-\tilde{V}_{ee}[G^{\rm bos}_{\varsigma}]$  \\\hline
       \multirow{6}{*}{$2$}        & 0.1  & 1.21  & 0.0563 & 0.01453 & 0.218 & -0.00742     \\\cline{2-7}
                                   & 0.3  & 1.79  & 0.0773 & 0.01126 & 0.243 & -0.00876     \\\cline{2-7}
                                   & 1    & 2.70  & 0.0604 & 0.00782 & 0.277 & -0.01147     \\\cline{2-7}
                                   & 3    & 3.78  & 0.0472 & 0.00662 & 0.339 & -0.00662     \\\cline{2-7}
                                   & 10   & 6.86  & 0.0026 & 0.00358 & 0.420 & -0.00605     \\\cline{2-7}
                                   & 100  & 52.1  & 0.0013 & 0.00097 & 0.667 & -0.00107     \\\hline
       \multirow{6}{*}{$3$}        & 0.1  & 0.67  & 0.0922 & 0.01850 & 0.814 & -0.01841     \\\cline{2-7}
                                   & 0.3  & 1.01  & 0.1146 & 0.02284 & 0.861 & -0.02098     \\\cline{2-7}
                                   & 1    & 1.42  & 0.1571 & 0.03816 & 0.932 & -0.03367     \\\cline{2-7}
                                   & 3    & 2.19  & 0.1867 & 0.02122 & 1.170 & -0.01880     \\\cline{2-7}
                                   & 10   & 7.02  & 0.1292 & 0.01245 & 1.438 & -0.01138     \\\cline{2-7}
                                   & 100  & 64.0  & 0.0388 & 0.00469 & 2.009 & -0.00738     \\\hline
       \multirow{6}{*}{$4$}        & 0.1  & 0.90  & 0.2408 & 0.03570 & 1.899 & -0.02522     \\\cline{2-7}
                                   & 0.3  & 1.42  & 0.2885 & 0.03687 & 2.092 & -0.04871     \\\cline{2-7}
                                   & 1    & 1.93  & 0.2809 & 0.05231 & 2.162 & -0.05275     \\\cline{2-7}
                                   & 3    & 9.48  & 0.3788 & 0.04266 & 2.773 & -0.05477     \\\cline{2-7}
                                   & 10   & 32.1  & 0.1024 & 0.02387 & 3.234 & -0.02013     \\\cline{2-7}
                                   & 100  & 232.0 & 0.0542 & 0.00681 & 3.751 & -0.00045     \\\hline
\end{tabular}
\caption{Approximations of the pair densities of homogeneous systems (for bosons).
The Coulomb energy in the last two columns is defined by $\tilde{V}_{ee}[\rho_2] = \int\int \rho_2(x,y)c(|x-y|)dxdy$.}
\label{table-appro-boson}
\end{table}
\end{small}

\begin{figure}[htb]
\centering
\includegraphics[width=4.5cm]{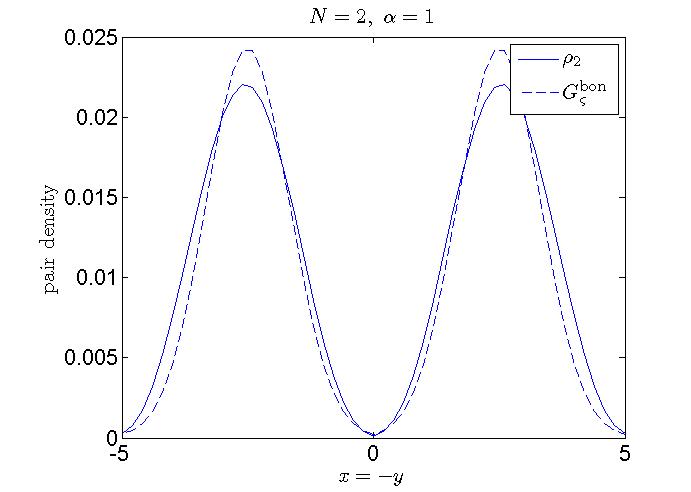}
\includegraphics[width=4.5cm]{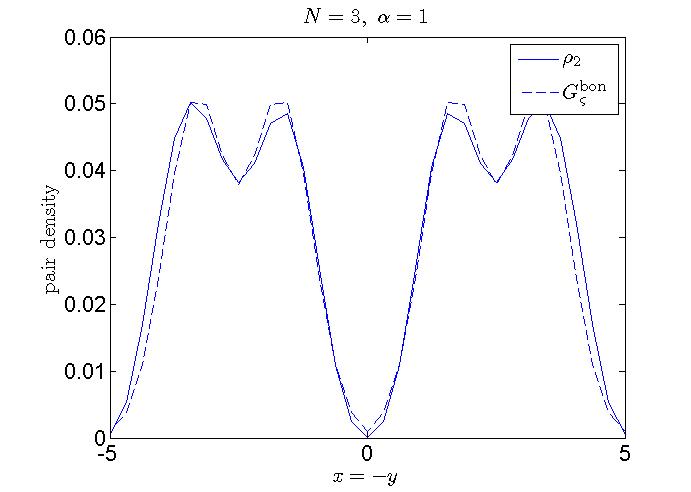}
\includegraphics[width=4.5cm]{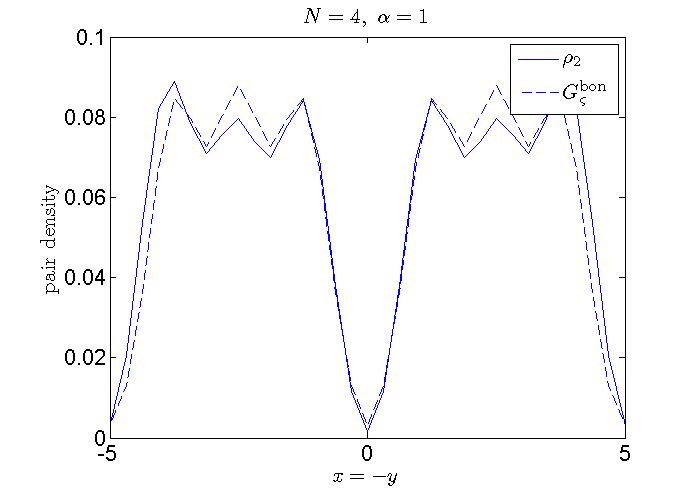}
\caption{\small The cross sections (on $x=-y$) of the pair densities
and their approximations $G^{\rm bos}_{\varsigma}$ for homogeneous electrons.}
\label{fig-cross-appro-boson}
\end{figure}

For fermions, to capture the asymptotic emergence of exact exchange as $\alpha\to\infty$ we make the ansatz
\begin{eqnarray}\label{ansatz-gaussian-HF}
G^{\rm fer}_{\varsigma,\eta}(x,y) = c_{\rm n} G^{\rm bos}_{\varsigma}(x,y)
\frac{1}{2}\left( \rho(x)\rho(y) - \eta\tau(x,y) \right),
\end{eqnarray}
where $c_n$ is a normalization constant, $\eta\in[0,1]$ is a parameter (allowed to depend on $N$ and $\alpha$), and $\tau$ is the exchange term from \eqref{exex'}. 
The freedom of varying $\eta$ allows a seamless crossover between the SCE pair density ($\eta=0,\,\varsigma=0$) and the exact-exchange pair density ($\eta=1,\,\varsigma=\infty$). 
The ansatz \eqref{ansatz-gaussian-HF} is not the only way to achieve this, but it is perhaps the simplest. Note that, unlike in B3LYP \cite{becke93}, exchange is mixed in {\it multiplicatively, not additively}. 
Numerically, we obtain $\eta$ by minimizing the $L^1$-error $\|\rho_2-G^{\rm fer}_{\varsigma,\eta}\|_{L^1(\Omega^2)}$ (while keeping, for simplicity, the bosonic values of $\varsigma$). The results in Table \ref{table-appro-fermion} and Figure \ref{fig-cross-appro-fermion} show that \eqref{ansatz-gaussian-HF} is a good approximation for fermions. In particular, Figure \ref{fig-cross-appro-fermion} (which concerns the case $N=4$ and different values of $\alpha$) shows that the transition from $6~(=2(N-1))$ SCE ridges to 4 exact-exchange ridges is correctly captured. 
The ansatz \eqref{ansatz-gaussian-HF} is accurate at the two limits $\alpha=0$ and $\alpha=\infty$, and the approximations are indeed better in the regimes where $\alpha$ is very small or large, as we can see from Table \ref{table-appro-fermion}.
Moreover, we observe that the errors for fermions are larger than those for bosons, which may be caused by the complicated interplay of Coulomb and exchange holes.

\begin{table}[ht]
\centering
\begin{tabular}{|c|c|c|c|c|c|c|}\hline
       $N$  &  $\alpha$  &  optimal $\eta$  &  $\|\rho_2-G^{\rm fer}_{\varsigma,\eta}\|_{L^1}$
       &  $\|\rho_2-G^{\rm fer}_{\varsigma,\eta}\|_{L^2}$  &    $\tilde{V}_{ee}[\rho_2]$
       &  $\tilde{V}_{ee}[\rho_2]-\tilde{V}_{ee}[G^{\rm fer}_{\varsigma,\eta}]$  \\\hline
       \multirow{6}{*}{$3$}        & 0.1  & 0    & 0.1375 & 0.02267 & 0.814 & -0.01694     \\\cline{2-7}
                                   & 0.3  & 0    & 0.2030 & 0.03601 & 1.416 & -0.01173     \\\cline{2-7}
                                   & 1    & 0.01 & 0.1919 & 0.02706 & 0.926 & -0.04983     \\\cline{2-7}
                                   & 3    & 0.02 & 0.2217 & 0.02681 & 1.094 & -0.04810     \\\cline{2-7}
                                   & 10   & 0.27 & 0.1792 & 0.02139 & 1.345 & -0.04138    \\\cline{2-7}
                                   & 100  & 0.92 & 0.0264 & 0.00324 & 1.676 & -0.00761     \\\hline
       \multirow{6}{*}{$4$}        & 0.1  & 0    & 0.3262 & 0.03189 & 1.898 & -0.04645     \\\cline{2-7}
                                   & 0.3  & 0    & 0.3222 & 0.03991 & 2.061 & -0.07957     \\\cline{2-7}
                                   & 1    & 0.01 & 0.3457 & 0.04296 & 2.133 & -0.07976     \\\cline{2-7}
                                   & 3    & 0.03 & 0.3539 & 0.03961 & 2.675 & -0.08353     \\\cline{2-7}
                                   & 10   & 0.47 & 0.1449 & 0.01752 & 2.971 & -0.06630     \\\cline{2-7}
                                   & 100  & 0.95 & 0.0223 & 0.00306 & 3.134 & -0.00155     \\\hline
\end{tabular}
\caption{Approximations of the pair densities of homogeneous systems (for fermions).}
\label{table-appro-fermion}
\end{table}

\begin{figure}[htb]
\centering
\includegraphics[width=4.5cm]{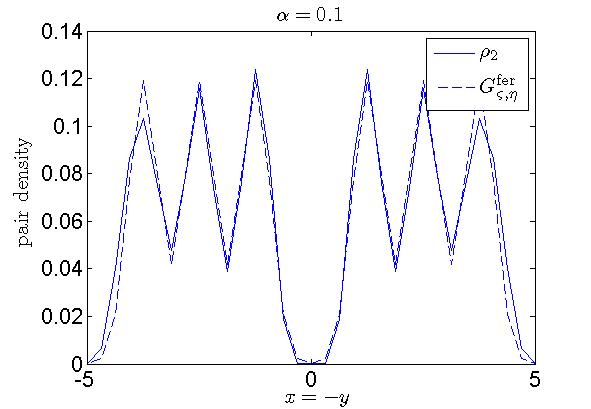}
\includegraphics[width=4.5cm]{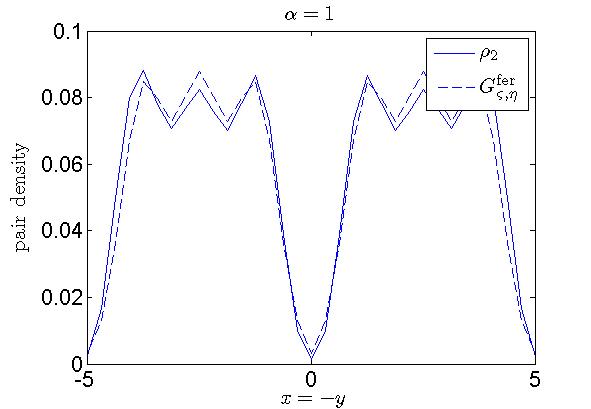}
\includegraphics[width=4.5cm]{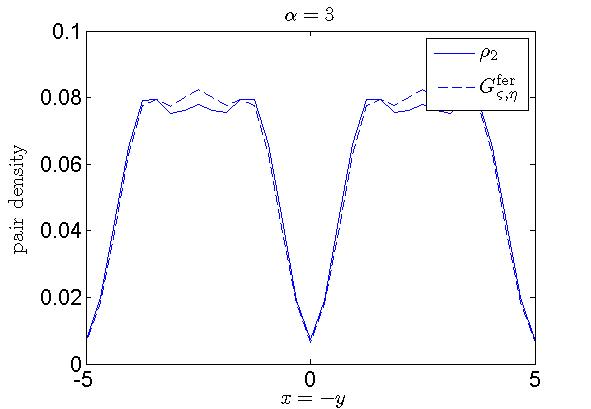}
\includegraphics[width=4.5cm]{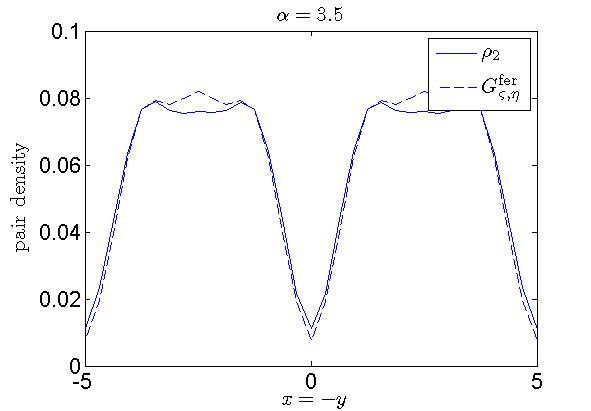}
\includegraphics[width=4.5cm]{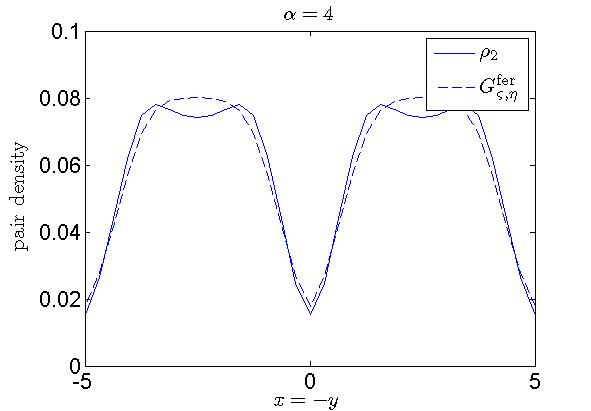}
\includegraphics[width=4.5cm]{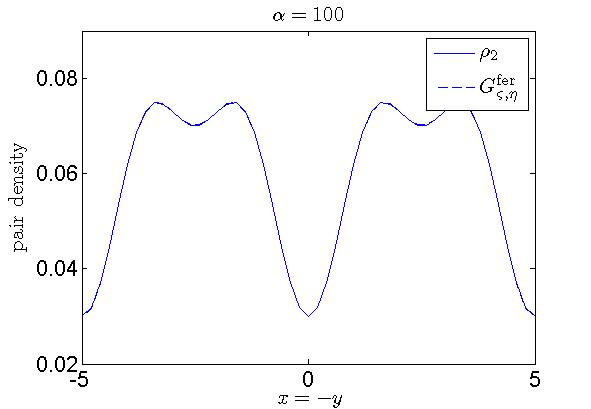}
\caption{\small The cross sections (on $x=-y$) of the pair densities
and their approximations $G^{\rm fer}_{\varsigma,\eta}$ for homogeneous systems with 4 electrons.}
\label{fig-cross-appro-fermion}
\end{figure}

\section{Conclusions}
\label{sec-conclusion} \setcounter{equation}{0}

In this paper we studied the exact density-to-pair-density map in density functional theory.
In the absence of any previous numerical simulations of this map, we
computed it here for typical one-dimensional families of densities obtained by scaling. This is the same as computing the map along the (two-sided) adiabatic connection from the non-interacting limit to the strictly correlated limit. We observed a slow and nontrivial cross-over between the endpoint profiles, which are given by exact exchange respectively by SCE correlations (or mathematically: by first-order perturbation theory respectively by optimal transport with Coulomb cost). The cross-over, while smooth, is very far from a linear interpolation and involves multiple lengthscales. 

This study gives us a deeper insight into the details of electron correlations, and may further lead to novel models for the pair density (and hence the interaction energy). 
As a fist step, we constructed an ansatz for pair densities of homogeneous systems in one dimension which is exact in the weak and the strong interaction limit and has been shown to remain accurate in the whole intermediate regime. The ansatz itself is readily generalized to inhomogeneous three-dimensional systems, but for such systems we have not yet tested its accuracy in the intermediate regime, nor do we know how to pick the correct parameter values just from the one-body density. We hope to come back to these issues in future work.
%
\\[4mm]
{\bf Acknowledgements} We thank Eric Canc\`es and Simen Kvaal for insightful comments on $v$-representability, and Andreas Savin for helpful discussions.

\end{document}